\documentclass[journal]{IEEEtran}

\usepackage{amsmath,amssymb,amsfonts}
\usepackage{bm}
\usepackage{cite}
\usepackage{array}
\usepackage{booktabs}
\usepackage{amsthm}
\usepackage{graphicx}
\usepackage{xcolor}
\usepackage[final]{microtype}
\usepackage{tikz}
\usetikzlibrary{positioning,arrows.meta,fit}
\usepackage{placeins}
\usepackage{url}
\allowdisplaybreaks[2]
\AtBeginDocument{%
\setlength{\abovedisplayskip}{2pt plus 0.5pt minus 0.5pt}%
\setlength{\belowdisplayskip}{2pt plus 0.5pt minus 0.5pt}%
\setlength{\abovedisplayshortskip}{1pt plus 0.5pt minus 0.5pt}%
\setlength{\belowdisplayshortskip}{1pt plus 0.5pt minus 0.5pt}%
}

\newtheorem{theorem}{Theorem}

\newtheorem{remark}{Remark}

\begin{document}
\raggedbottom
\bstctlcite{IEEEtranBSTCTLNoDash}

\title{Reshaping Converter–Network Interactions in Microgrids: From Virtual Impedance to Virtual Two-Port Control}

\author{Liaoyuan Yang,
        Peng Yang,~\IEEEmembership{Member,~IEEE},
        Yang Wu,~\IEEEmembership{Member,~IEEE},
        and Feng Liu,~\IEEEmembership{Senior Member,~IEEE}%
\thanks{This work was supported by the National Natural Science Foundation of China under Grant 52407137.
(Corresponding author: Peng Yang.)}%
\thanks{P. Yang and L. Yang are with the School of Electronics and
Information, Xi'an Polytechnic University, Xi'an 710048, China (e-mail:
p-yang13@tsinghua.org.cn).}%
\thanks{Y. Wu is with the School of Electrical and Electronic Engineering,
Nanyang Technological University, Singapore 639798 (e-mail:
yangwu@ntu.edu.sg).}%
\thanks{F. Liu is with the Department of Electrical Engineering, Tsinghua
University, Beijing 100084, China (e-mail: lfeng@tsinghua.edu.cn).}%
}

\maketitle

\begin{abstract}
Converter terminal characteristics are central to both
dynamic interactions with the network and steady-state power
sharing in inverter-based microgrids. Conventional additive virtual
impedance (VI) shapes these characteristics through a single
virtual branch. This paper proposes virtual two-port control,
which uses four coordinated transfer channels to reconstruct
the terminal impedance. The connected network consequently
observes the original converter through a virtually inserted twoport interface, with conventional virtual impedance recovered as
a degenerate case. The proposed control transforms the original
impedance through a matrix linear-fractional map. We also derive
a necessary and sufficient condition for the reconstruction to
be stable and proper. Among its various potential applications in microgrids,
we develop three representative ones in detail: passivation with
reduced control effort, uncertainty compression, and seriesshunt power-flow regulation. Simulations of VSG-controlled
converters demonstrate the advantages of the proposed method
over conventional VI in these applications.
\end{abstract}


\section{Introduction}

With the increasing integration of renewable energy resources, voltage-source
converters (VSCs) have become a major interface between power-electronic
generation units and the power grid \cite{Xiong2022VSCReview}. The growing
penetration of VSC-interfaced resources is shifting modern power grids toward
converter-dominated systems \cite{Henderson2024GSIM}, with inverter-based
microgrids being an important setting in which multiple converters interact
through the electrical network\cite{Watson2021ScalableGFM}. Since every converter interacts with the
surrounding network through terminal voltage and current, shaping this local
relation provides a means of influencing interactions among converters,
feeders, and loads without redesigning the entire microgrid controller.

This terminal relation matters in both dynamic and steady-state operation. From the dynamic perspective, it is commonly studied as the
converter output impedance in the frequency domain
\cite{Sun2011Impedance}. Extensive studies have developed impedance modeling
and analysis methods for VSCs
\cite{Cespedes2014Impedance,Wang2018UnifiedImpedance} and revealed how the
output impedance and terminal characteristics affect converter-network
interactions and stability
\cite{Zhang2020FrequencyCoupling,Cao2020TerminalCharacteristics}. From the
steady-state perspective, the same terminal relation contributes to the
effective impedances between converters, loads, and microgrid buses, thereby
influencing power sharing and power-flow redistribution
\cite{Han2016Review,Deng2023ImpedanceReshaping}, especially in high-$R/X$
networks \cite{Guerrero2007ResistiveOutput}.
These two viewpoints lead to the same design question: how can a converter
controller reshape the terminal relation governing its interaction with the
microgrid?

Virtual impedance (VI) is widely used to shape this interface in
inverter-based microgrids. Early output-impedance designs improved load
sharing and suppressed circulating currents in parallel inverters
\cite{Guerrero2005UPS}. The realization and the
range of artificial resistive, inductive, capacitive, and frequency-dependent
impedances were subsequently formalized
\cite{He2011VirtualImpedance,Wang2015VirtualImpedance}. Later extensions used
the same virtual-branch concept for current limiting
\cite{Qoria2020CurrentLimiting,Baeckeland2024OvercurrentReview}, adaptive or iterative sharing correction
\cite{Vijay2021AdaptiveVI,An2021SuccessiveApproximation}, asymmetric and
sequence-selective shaping \cite{Jin2022AsymmetricalVI,Kim2018NegativeSequenceVI},
harmonic control \cite{Goethner2021HarmonicVI}, weak-grid impedance reshaping
\cite{Guo2024SeriesReshaping}, and operating-condition-dependent protection
\cite{Qoria2023VariableVI}. 

Related approaches also shape terminal behavior through virtual admittance
and coordinated feedback/feedforward paths. Virtual-admittance-based
grid-forming control provides voltage and frequency support for
current-controlled converters \cite{VidalLeon2023VirtualAdmittance}.
Impedance circuit models reveal the contributions of control loops,
feedforward, and decoupling paths \cite{Li2021ImpedanceCircuit}.
Two-degree-of-freedom control based on uncertainty and disturbance
estimation also directly shapes output impedance to improve inverter
voltage quality under nonlinear loads \cite{Gadelovits2017UDE}.

For the fixed additive VI considered here, the terminal action remains
series-like. A Thevenin--Norton source conversion preserves the same
one-port relation.
Existing studies have shown that VI design must account for current limits,
controller saturation, and control effort
\cite{Jafariazad2024AdaptiveVI,DehghaniArani2024FRT}. In practice, once the
required virtual action approaches the converter voltage or current limits,
the implemented terminal behavior may depart from that assumed in the design,
reducing the accuracy of microgrid response predictions based on the
designed terminal characteristic. Converter interaction characteristics are also affected by parameter
variations, delays, and operating conditions
\cite{Wang2022PassivityEnhancement,Cecati2025Interoperability}, while a fixed
additive VI applies the same branch to every member of the impedance family. At the
network level, the same VI setting simultaneously influences power sharing,
voltage regulation, and other microgrid objectives
\cite{Deng2023ImpedanceReshaping,Wu2017VirtualImpedanceRange,
Xiao2024VirtualImpedance}. Hence, uncertainty accommodation and power-flow
correction remain coupled to the same series-like action, which can restrict
the terminal behaviors that a converter can reliably realize for coordinated
microgrid operation.

To address this issue, this
paper extends virtual impedance to a more general terminal-control structure,
referred to as virtual two-port control. Its four-block voltage/current transformation gives an explicit
impedance map and, for a reciprocal subclass, a series--shunt network
interpretation, connecting dynamic interaction shaping with steady-state
network regulation. The converter terminal is reconstructed through
coordinated feedback and feedforward channels, with VI as a degenerate
special case. Our preliminary work introduced virtual two-port control for symmetric
and asymmetric impedance shaping and series--shunt power-flow regulation
\cite{YangVTPC2026}. This paper develops stable-proper return conditions,
coefficient-gain bounds, passivity-margin compression guarantees, and
positive-resistance voltage-reachability conditions, supported by
VSG-based case studies. Among the broader potential
applications of the proposed control, we select three representative applications for detailed analysis and simulations. The
contributions are as follows.

\begin{itemize}

\item We propose a four-block terminal-reconstruction law that emulates a
virtual two-port between the original VSC and the network. We show that the resulting
controlled impedance is a matrix linear-fractional transformation of the
original one. We also establish conditions ensuring that the controlled impedance
is stable and proper.

\item Three representative applications demonstrate dynamic interaction
shaping and steady-state network regulation through the same port
transformation in inverter-based microgrids. Specifically, virtual
two-port control can:
1) satisfy a prescribed passivity requirement with a lower controller
coefficient gain, reducing the internal command burden and avoiding
undesirable saturation; 2) compress uncertainty-induced variation in the
terminal impedance characteristic, as measured by the passivity-margin
spread, reducing the terminal variation presented to the microgrid; and
3) extend power-flow regulation beyond series-only compensation,
improving power-sharing accuracy in microgrids with smaller voltage drops.

\end{itemize}

The remainder of this paper is organized as follows. Section II develops the
proposed control with stable-proper condition.
Sections III-V develop the three applications with theoretical analysis and case
studies, respectively. Section VI concludes the paper.

Notations: \(\mathcal{RH}_{\infty}^{m\times n}\) denotes the set of
real-rational, stable, proper \(m\times n\) transfer matrices. For a matrix $X$, $\operatorname{He}\{X\}:=(X+X^{\mathrm H})/{2}$ denotes its Hermitian part.

\section{Proposed Virtual Two-Port Control}
\label{sec:virtual_two_port_impedance}

This section develops the proposed controller from the terminal relation
observed by the connected network. We first
introduce the four-block
voltage-current reconstruction and relate it to two-port
parameters and VI. Then, we
derive the resulting impedance map and the condition under which the reconstructed impedance is stable and proper. We also state a
subclass of the proposed control that admits a reciprocal \(\Pi\)-network interpretation.

\subsection{Control Structure}
\label{subsec:z_parameter_two_port}
At a converter terminal, let \(u(s)\) and \(i(s)\) denote the incremental \(dq\) voltage
and current, respectively, with \(i(s)\) defined positive toward the
converter:
\begin{equation*}
u(s)=
\begin{bmatrix}
\Delta U_d(s)\\
\Delta U_q(s)
\end{bmatrix}
\in\mathbb C^2,
\qquad
i(s)=
\begin{bmatrix}
\Delta I_d(s)\\
\Delta I_q(s)
\end{bmatrix}
\in\mathbb C^2.
\end{equation*}
The terminal impedance of the converter is then defined by
\begin{equation}
    u(s)=Z(s)i(s),
    \label{eq:original_impedance}
\end{equation}
where
\begin{equation*}
Z(s)=
\begin{bmatrix}
Z_{dd}(s) & Z_{dq}(s)\\
Z_{qd}(s) & Z_{qq}(s)
\end{bmatrix}
\in\mathbb C^{2\times2}.
\end{equation*}
is the original terminal impedance to be reshaped. For notational simplicity, the Laplace variable \(s\) is omitted from transfer functions when clear from context.

Our central idea is to virtually insert a controllable two-port interface
between the original converter impedance \(Z\) and the connected network, as
illustrated in Fig.~\ref{fig:virtual_two_port_concept}. The internal original port variables \((u,i)\) satisfy relation \(u=Zi\),
whereas the controlled external port variables \((\tilde u,\tilde i)\) satisfy a new relation \(\tilde u=\tilde Z \tilde i\) to be designed. By controlling the relation
between these two pairs of port variables, the interface reshapes the terminal
behavior observed by the network. 

\begin{figure}[!htbp]
\centering
\includegraphics[width=0.65\linewidth]{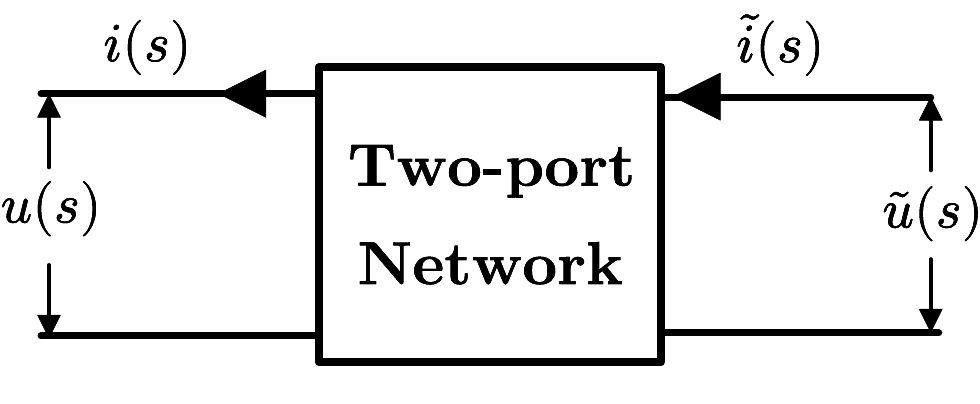}
\caption{Virtual two-port interface inserted between the original converter
impedance and the external network.}
\label{fig:virtual_two_port_concept}
\end{figure}

To realize this idea as a controller, we parameterize the desired relation
between the internal and external port variables by four transfer matrices:
\begin{equation}
\left\{
\begin{aligned}
 \tilde u(s) &= A(s)i(s)+B(s)u(s),\\
 i(s) &= C(s)\tilde i(s)+D(s)\tilde u(s),
\end{aligned}
\right.
\label{eq:twoport_implementation_ABCD}
\end{equation}
where \(A,B,C,D\in\mathcal{RH}_{\infty}^{2\times2}\) are selected according
to the terminal-shaping objective. The four channels have direct port-level interpretations. The term $Ai$
provides a current-dependent voltage injection and therefore retains the
series-like action of conventional virtual impedance. The channel $B$
determines how the original converter voltage is transferred to the
network-side port. On the current side, $C$ determines how the network current
is transmitted to the original converter port, whereas $D\tilde u$ introduces
a voltage-dependent current component with a shunt-like interpretation.
Therefore, the interface is no longer restricted to inserting an additional
voltage drop along the original current path; it can reshape both voltage and
current transfer between the converter and the network.

Equation \eqref{eq:twoport_implementation_ABCD} leads directly to the control block diagram as shown in
Fig.~\ref{fig:io_implementation}. The controller emulates the port
behavior of an inserted two-port network through signal reconstruction.

\begin{figure}[!htbp]
\centering
\includegraphics[width=0.70\linewidth]{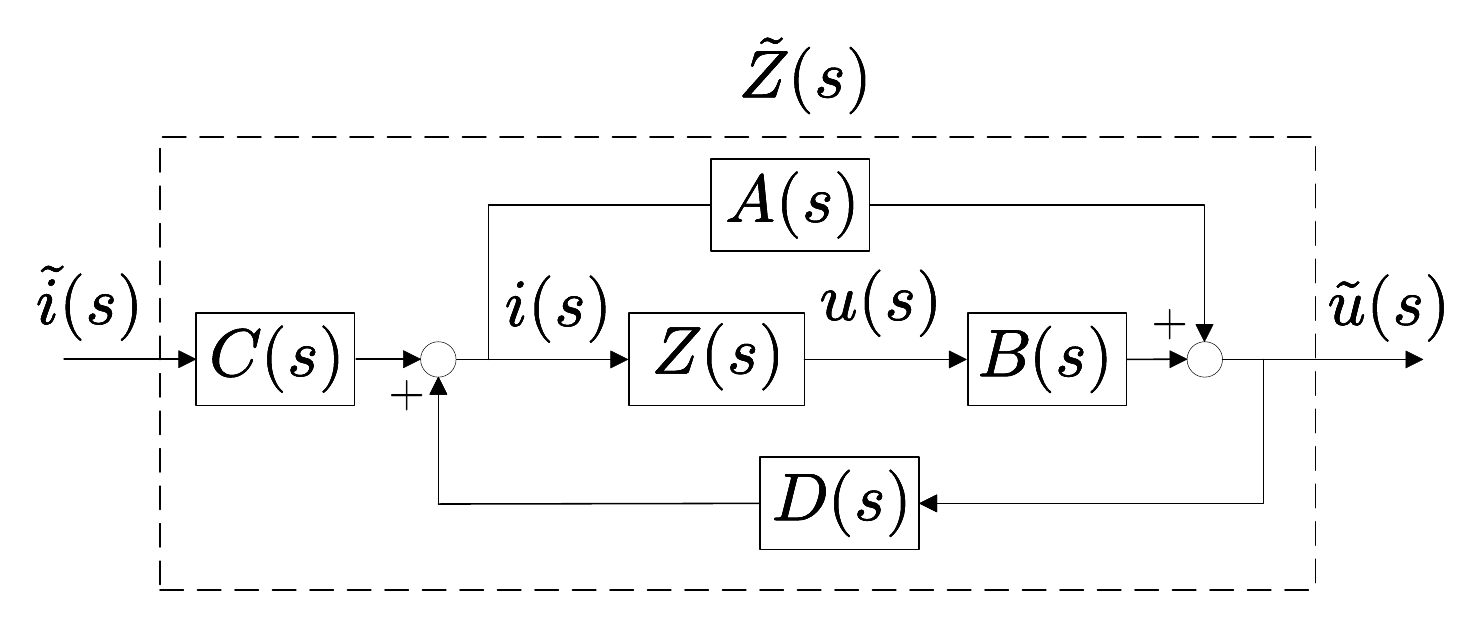}
\caption{Four-block implementation of the virtual two-port control.}
\label{fig:io_implementation}
\end{figure}

\begin{remark}[Equilibrium preservation]
\label{rem:equilibrium_compensation}
Equation~\eqref{eq:twoport_implementation_ABCD} is written for incremental
variables about a selected equilibrium. Let
\(u_{\rm f}\), \(i_{\rm f}\), \(\tilde u_{\rm f}\), and
\(\tilde i_{\rm f}\) denote the corresponding full signals, and let
\((u_0,i_0,\tilde u_0,\tilde i_0)\) denote their equilibrium values. The
controller can be implemented as
\begin{equation*}
\begin{aligned}
\tilde u_{\rm f}-\tilde u_0
&=
A(s)(i_{\rm f}-i_0)+B(s)(u_{\rm f}-u_0),\\
i_{\rm f}-i_0
&=
C(s)(\tilde i_{\rm f}-\tilde i_0)
+D(s)(\tilde u_{\rm f}-\tilde u_0).
\end{aligned}
\end{equation*}
Equivalently, if the full signals are applied directly to the four blocks,
the required constant compensation terms are
\begin{equation*}
\begin{aligned}
k_u&=\tilde u_0-A(0)i_0-B(0)u_0,\\
k_i&=i_0-C(0)\tilde i_0-D(0)\tilde u_0.
\end{aligned}
\end{equation*}
With the blocks initialized at their steady states, these offsets preserve
the selected equilibrium and the incremental impedance map. Applications~1
and~2 use this equilibrium-preserving form. Application~3 uses the full
steady-state port relation to reshape power flow, with incremental
dynamics defined about the resulting equilibrium.
\end{remark}

The virtual impedance control is a special case by setting
\begin{equation*}
    A(s)=Z_v(s),\qquad B(s)=C(s)=I,\qquad D(s)=0,
    \label{eq:vi_degenerate_case}
\end{equation*}
where $Z_v(s)$ is the virtual impedance. This gives \(\tilde u=u+Z_v i\) and \(\tilde i=i\), and the controlled
impedance becomes $\tilde Z(s)=Z(s)+Z_v(s)$. 

The connection between the proposed control structure to the classical two-port network equation is also explicit. With the
current directions in Fig.~\ref{fig:virtual_two_port_concept}, the currents
entering the virtual interface are \(-i\) and \(\tilde i\), so a finite
impedance-parameter representation of the two-port network is
\begin{equation}
 \begin{bmatrix}
  u(s)\\
  \tilde u(s)
 \end{bmatrix}
 =
 \begin{bmatrix}
  Z_{11}(s) & Z_{12}(s)\\
  Z_{21}(s) & Z_{22}(s)
 \end{bmatrix}
 \begin{bmatrix}
  -i(s)\\
  \tilde i(s)
 \end{bmatrix}.
 \label{eq:z_parameter_two_port}
\end{equation}
When \(Z_{12}\) and \(Z_{21}\) are nonsingular, elimination of the port
variables yields \eqref{eq:twoport_implementation_ABCD} with
\begin{equation}
\left\{
\begin{aligned}
 A &=-Z_{21}+Z_{22}Z_{12}^{-1}Z_{11},\\
 B &=Z_{22}Z_{12}^{-1},\\
 C &=Z_{21}^{-1}Z_{22},\\
 D &=-Z_{21}^{-1}.
\end{aligned}
\right.
\label{eq:ABCD_from_z_parameters}
\end{equation}
We remark that the controller representation \eqref{eq:twoport_implementation_ABCD} is more
convenient for implementation and also includes singular or limiting cases,
such as VI, that are not described by finite cross-parameter inverses.

\subsection{Linear-Fractional Impedance Map}
\label{subsec:equivalent_impedance_closed_loop}

Connecting the reconstruction law to the original terminal relation
\eqref{eq:original_impedance} gives
\[
 \tilde u=(A+BZ)i.
\]
Define
\begin{equation}
 H:=A+BZ,
 \qquad
 S:=(I-DH)^{-1}.
 \label{eq:stable_return_inverse}
\end{equation}
The current-reconstruction equation then gives
\[[I-DH]i=C\tilde i\]. 
Whenever the inverse in
\eqref{eq:stable_return_inverse} exists, we have
\(i=SC\tilde i\). And hence the controlled impedance is
\begin{equation}
 \widetilde Z
 =(A+BZ)[I-D(A+BZ)]^{-1}C.
 \label{eq:closed_loop_equivalent_impedance_alt}
\end{equation}

This shows that the virtual two-port control acts on the original impedance through a
matrix linear-fractional transformation. The affine term \(A+BZ\) sets the
forward transformation, \(D\) creates the return denominator, and \(C\)
postprocesses the reconstructed current. The two-port control forms a return interconnection through
\(D\) and \(H=A+BZ\). Consequently, stability and properness of the individual blocks alone do not ensure a stable and proper reconstructed terminal impedance; the return inverse
\(
(I-DH)^{-1}
\)
must also be considered. The following theorem characterizes when this
return inverse, and hence the reconstructed terminal impedance, is stable
and proper.

\begin{theorem}[Stable-proper condition]
\label{thm:stable_reconstruction}
Let the original terminal impedance and the four control blocks satisfy
\(
Z,A,B,C,D\in\mathcal{RH}_{\infty}^{2\times2}.
\)
Then \(S=(I-DH)^{-1}\) belongs to
\(\mathcal{RH}_{\infty}^{2\times2}\) if and only if $$\det[I-D(\infty)H(\infty)]\neq0$$ 
and 
$$\det[I-D(s)H(s)]\neq0, \quad \forall s:\ \operatorname{Re}(s)\geq0.$$
When these conditions hold, the controlled impedance satisfies
\(
\widetilde Z=HSC
\in\mathcal{RH}_{\infty}^{2\times2}.
\)
\end{theorem}

\noindent\emph{Proof:} See
Appendix~\ref{app:proof_stable_reconstruction}.

Theorem~\ref{thm:stable_reconstruction} identifies the return channel \(DH\) as the source of the additional poles introduced by the two-port
reconstruction. The blocks \(A\), \(B\), and \(Z\) determine \(H\), and \(D\)
closes the return path; therefore, the additional finite poles are governed by
the zeros of \(\det[I-D(s)H(s)]\). The block \(C\) changes
\(\widetilde Z=HSC\) but does not enter this denominator. A
convenient sufficient design condition is
\begin{equation}
 \|DH\|_\infty<1,
 \label{eq:return_small_gain_condition}
\end{equation}
which guarantees \(S\in\mathcal{RH}_\infty\). Alternatively, the exact
conditions in Theorem~\ref{thm:stable_reconstruction} can be verified directly.

When \(D=0\), the map reduces to the affine form
\(\widetilde Z=(A+BZ)C\), and no return poles are introduced. VI is the
further specialization \((A,B,C)=(Z_v,I,I)\), which gives
\(\widetilde Z=Z+Z_v\). A nonzero \(D\) supplies the genuinely
linear-fractional shaping freedom, together with the return condition in
Theorem~\ref{thm:stable_reconstruction}.

\subsection{Reciprocal $\Pi$-Network Realization}
\label{subsec:reciprocal_pi_equivalent}

The general four-block controller need not admit a circuit realization. However, a
reciprocal subclass of the proposed control can provide a direct branch interpretation, which may be useful in network-level design. Consider the virtual $\Pi$ network in
Fig.~\ref{fig:virtual_pi_network}, with converter-side shunt impedance $Z_1$,
series impedance $Z_2$, and line-side shunt impedance $Z_3$. Define
$Y_k:=Z_k^{-1}$. Its nodal equation is
\begin{equation}
 \begin{bmatrix}
  -i\\
  \tilde i
 \end{bmatrix}
 =
 \begin{bmatrix}
  Y_1+Y_2 & -Y_2\\
  -Y_2 & Y_2+Y_3
 \end{bmatrix}
 \begin{bmatrix}
  u\\
  \tilde u
 \end{bmatrix}.
 \label{eq:pi_admittance_two_port}
\end{equation}

\begin{figure}[!htbp]
 \centering
 \includegraphics[width=0.92\columnwidth]{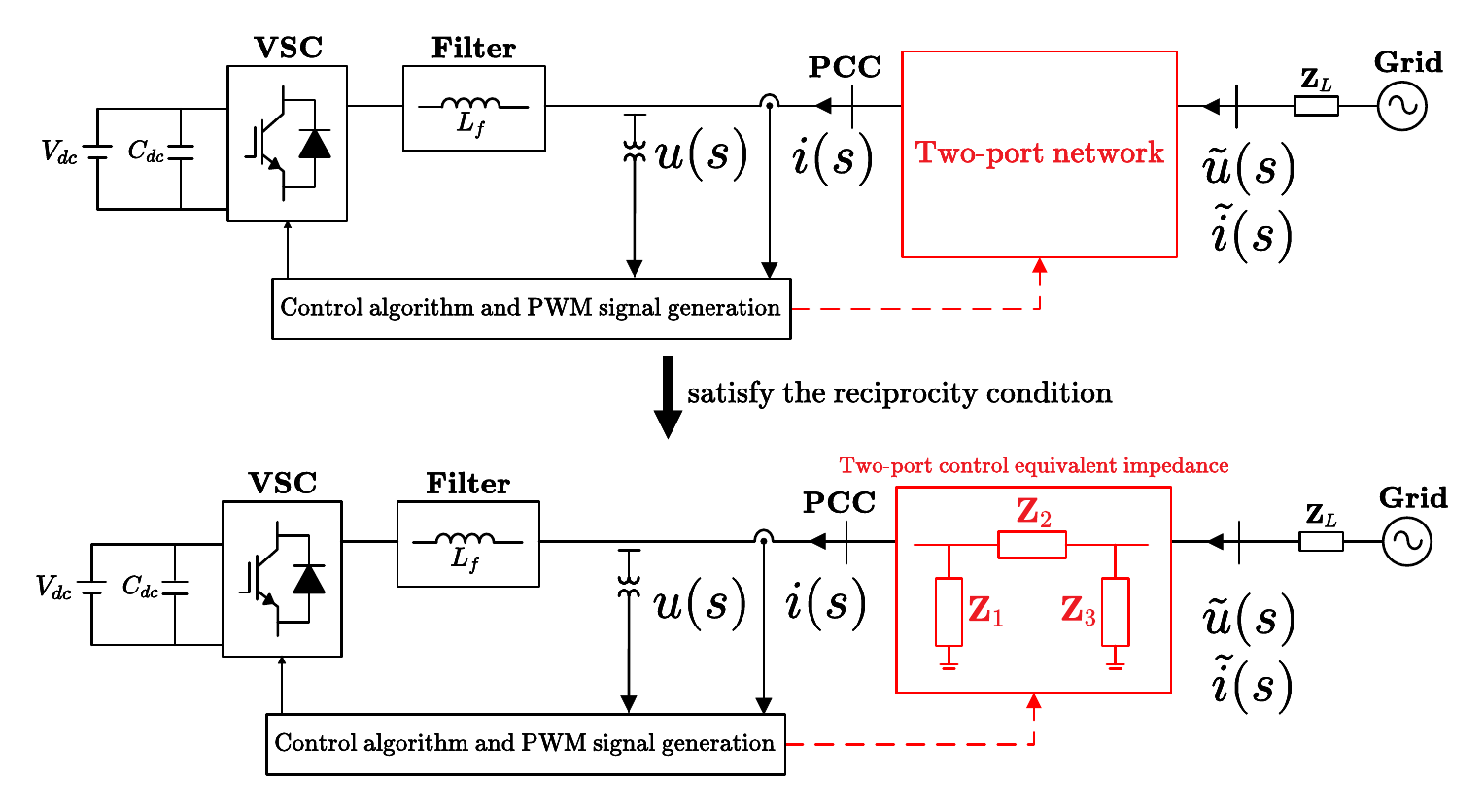}
 \caption{Reciprocal virtual $\Pi$ network and its terminal variables.}
 \label{fig:virtual_pi_network}
\end{figure}

Solving \eqref{eq:pi_admittance_two_port} in the form of
\eqref{eq:twoport_implementation_ABCD} gives
\begin{equation}
\left\{
\begin{aligned}
 A &= Z_2,\\
 B &= I+Z_2Y_1,\\
 C &= I+Y_1Z_2,\\
 D &= -(Y_1+Y_3+Y_1Z_2Y_3).
\end{aligned}
\right.
\label{eq:reciprocal_pi_abcd_blocks}
\end{equation}
The full conversion between impedance parameters and finite reciprocal
$\Pi$ branches, together with its nonsingularity and uniqueness conditions,
is given in Appendix~\ref{app:reciprocal_pi_details}.

The network interpretation is clearest when the converter-side shunt is open,
i.e., $Y_1=0$. Equation~\eqref{eq:reciprocal_pi_abcd_blocks} then reduces to
\begin{equation}
 A=Z_2,\qquad B=C=I,\qquad D=-Y_3,
 \label{eq:open_converter_shunt_abcd}
\end{equation}
or, equivalently,
\begin{equation}
 \tilde u=u+Z_2i,\qquad
 \tilde i=i+Y_3\tilde u.
 \label{eq:series_shunt_terminal_reconstruction}
\end{equation}
The series branch $Z_2$ reconstructs terminal voltage, while the shunt branch
$Y_3$ reconstructs terminal current. Setting $Y_3=0$ and $Z_2=Z_v$ recovers
conventional VI exactly. A nonzero $Y_3$ therefore supplies the additional
shunt degree of freedom that independently modifies a network self-admittance;
Section~\ref{subsec:nonnegative_resistance_capability} uses this structured
subset for voltage reachability and network-symmetry regulation.

The general \(A,B,C,D\) construction and its stable-proper return-inverse
condition apply to all three applications. Applications 1 and 2 use the
unrestricted reconstruction degrees of freedom for dynamic impedance shaping,
whereas Application 3 uses the reciprocal series-shunt specialization for
network-level voltage and current coordination.

\section{Application 1: Passivation Control in Microgrid with Less Control Effort}
\label{sec:finite_gain_passivation}

\subsection{Design Objective}

In a microgrid, each converter interacts with neighboring
units and feeders through its terminal voltage-current relation. Passivity-oriented impedance shaping provides a local means of mitigating
adverse converter-network interactions \cite{Miranbeigi2022PassivityVI}.
Against this background, the first application asks a practical design
question: when the original converter impedance is nonpassive, can the
required passivity margin be reached with smaller controller coefficients
than conventional VI control? This matters because large coefficients can
generate large inner-loop commands that activate converter limits and make
the implemented terminal behavior depart from its design
\cite{Baeckeland2024OvercurrentReview,Jafariazad2024AdaptiveVI}.

To compare the two structures on the same basis, both are required
to achieve the same passivity margin over the same frequency set.
Let \(\Omega \subset \mathbb{R}_{\geq 0}\) denote that angular-frequency set.
For a stable proper terminal impedance \(Z\), its passivity margin over
\(\Omega\) is defined as
\begin{equation}
\mu(Z;\Omega)
 :={}
 \inf_{\omega\in\Omega}
 \lambda_{\min}
 \left[\operatorname{He}\{Z(\mathrm{j}\omega)\}\right],
 \label{eq:passivity_margin_definition}
\end{equation}
A prescribed target \(\varepsilon>0\) is attained when
\(\mu(Z;\Omega)\geq\varepsilon\). The comparison also requires a common
measure of controller magnitude. All coefficients below are expressed in per
unit on fixed converter bases. Relative to the uncontrolled identity interface
\((A,B,C,D)=(0,I,I,0)\), define
\begin{equation}
 \gamma:=
 \left\|
 \begin{bmatrix}
 A&B-I\\
 C-I&D
 \end{bmatrix}
 \right\|_\infty.
 \label{eq:twoport_dimensionless_gain}
\end{equation}
VI control is the degenerate choice
\((A,B,C,D)=(Z_v,I,I,0)\). Hence the same definition reduces to
\(\gamma=\|Z_v\|_\infty\), allowing both controller structures to be
compared with one measure. The analysis first establishes the smallest
possible VI gain and then constructs a virtual two-port control that can improve on this
benchmark.

\subsection{VI Benchmark: Minimum Passivation Gain}

Additive VI gives \(\widetilde Z=Z+Z_v\) and shifts the Hermitian part of
the terminal impedance directly. The question is how large this additive
coefficient must be to attain \(\varepsilon\). The following theorem provides
the VI benchmark for the two-port comparison.

\begin{theorem}[Minimum VI gain for passivation]
\label{thm:finite_gain_vi_limit}
Let \(Z\in\mathcal{RH}_{\infty}^{2\times2}\), let
\(\mu_0:=\mu(Z;\Omega)\), and prescribe
\(\varepsilon>\max\{0,\mu_0\}\). Then the minimal gain for VI control, denoted by $\gamma_{\rm VI}^{\star}$, is \footnote{Subscripts “VI” and “TP” are
used to distinguish variables associated with different controls.}
\begin{equation}
 \gamma_{\rm VI}^{\star}
 :=
 \inf_{\substack{
 Z_v\in\mathcal{RH}_{\infty}^{2\times2}\\
 \mu(Z+Z_v;\Omega)\geq\varepsilon}}
 \|Z_v\|_\infty
 =\varepsilon-\mu_0.
 \label{eq:minimum_vi_gain}
\end{equation}
\end{theorem}

\noindent\emph{Proof:} See Appendix~\ref{app:proof_minimum_vi_gain}.

The theorem establishes the limitation of additive VI in this application:
the passivity shortage \(\varepsilon-\mu_0\) appears directly as the minimum
coefficient gain. A larger shortage therefore requires a larger reconstruction coefficient and
can amplify the commands demanded from the retained inner loops. Consequently,
stronger passivity-oriented terminal shaping through additive VI places a
proportionally larger implementation burden on the local converter.

\subsection{Two-Port Design with Reduced Coefficient Gain}

The VI benchmark shows that an additive controller requires a coefficient gain
equal to the full passivity shortage. We next ask whether the virtual two-port can attain the same prescribed margin with a smaller
coefficient gain. The following theorem gives a constructive design and the
exact condition under which it improves on the minimum VI benchmark.
\begin{theorem}[Two-port passivation and gain reduction]
\label{thm:reduced_gain_tp}
Under the conditions of Theorem~\ref{thm:finite_gain_vi_limit}, choose any
scaling factor \(0<\rho<1\) and set
\begin{equation}
 A=(\varepsilon-\rho\mu_0)I,\qquad
 B=\rho I,\qquad C=I,\qquad D=0.
 \label{eq:reduced_gain_tp_blocks}
\end{equation}
This realization has
\(S=I\in\mathcal{RH}_{\infty}^{2\times2}\) and gives
\begin{equation}
 \widetilde Z_\rho=(\varepsilon-\rho\mu_0)I+\rho Z,
 \qquad
 \mu(\widetilde Z_\rho;\Omega)=\varepsilon,
 \label{eq:reduced_gain_tp_impedance}
\end{equation}
with coefficient gain
\begin{equation}
 \gamma_{\rm TP}
 =\sqrt{(\varepsilon-\rho\mu_0)^2+(1-\rho)^2}.
 \label{eq:reduced_gain_tp_gain}
\end{equation}
Compared with the minimum VI gain, the strict advantage $\gamma_{\rm TP}<\gamma_{\rm VI}^{\star}$
holds if and only if
\begin{equation}
 (1+\rho)\mu_0^2-2\varepsilon\mu_0+\rho-1>0.
 \label{eq:strict_gain_condition}
\end{equation}
\end{theorem}

\noindent\emph{Proof:} See Appendix~\ref{app:proof_reduced_gain_tp}.

Note that, since VI is the special
choice \((A,B,C,D)=(Z_v,I,I,0)\), every feasible VI design is also a feasible
two-port design, and therefore the two-port structure cannot perform worse than VI in terms of the minimum achievable coefficient gain.
Indeed, the boundary choice \(\rho=1\) in
\eqref{eq:reduced_gain_tp_blocks} recovers the minimum-gain VI realization
\(A=(\varepsilon-\mu_0)I\), \(B=C=I\), and \(D=0\). Theorem~\ref{thm:reduced_gain_tp} establishes the stronger
result: it identifies when a nondegenerate choice \(0<\rho<1\) achieves a strict improvement over the VI optimum.
Under the theorem assumptions,
\(\mu_0<\varepsilon\), and \eqref{eq:strict_gain_condition} is equivalent to
\begin{equation}
  \mu_0<
  \frac{\varepsilon-\sqrt{\varepsilon^2+1-\rho^2}}{1+\rho}.
\end{equation}
Thus, this construction improves on the optimal VI when the original
impedance has a sufficiently negative worst-case passivity margin. 
The advantage becomes especially transparent for a strongly nonpassive
original impedance. As
\(\mu_0\to-\infty\),
\begin{equation}
 \frac{\gamma_{\rm TP}}
 {\gamma_{\rm VI}^{\star}}
 =
 \frac{\sqrt{(\varepsilon-\rho\mu_0)^2+(1-\rho)^2}}
 {\varepsilon-\mu_0}
 \longrightarrow \rho<1.
 \label{eq:severe_shortage_gain_ratio}
\end{equation}
Hence \eqref{eq:strict_gain_condition} necessarily holds when the original
passivity shortage is sufficiently large, and the asymptotic coefficient
gain is only a fraction \(\rho\) of the minimum required by VI. Physically, the two-port does not rely entirely on an additional
current-dependent voltage drop to compensate for the original nonpassive
characteristic. By also changing how the original converter characteristic is
transferred to the external port, it can meet the same local
interaction-shaping objective with a smaller reconstruction effort.

\subsection{Case Study}
\label{subsec:case1_finite_gain}

The preceding analysis shows that VI and two-port control can achieve the
same pointwise passivity target, but with different coefficient gains. With
finite current-reference and converter-voltage limits, this gain difference
may determine whether the required commands can be applied and whether the
implemented terminal response remains consistent with its unlimited design.
To examine how this local implementation effect propagates through a
microgrid, consider the grid-connected system in
Fig.~\ref{fig:case1_network}, with VSG~1 at the PCC, a local constant-current
load, and VSG~2 connected through a feeder. The PCC is connected to the
utility through a second feeder. Both converters use continuous-time averaged
models with VSG outer loops, cascaded voltage/current PI loops,
filter-current dynamics, command limits, and anti-windup. VI or two-port
control is applied only to VSG~1 at the outer-loop interface in
Fig.~\ref{fig:case_common_model}; VSG~2 retains its conventional VSG control.

\begin{figure}[!htbp]
  \centering
  \includegraphics[width=0.63\columnwidth]{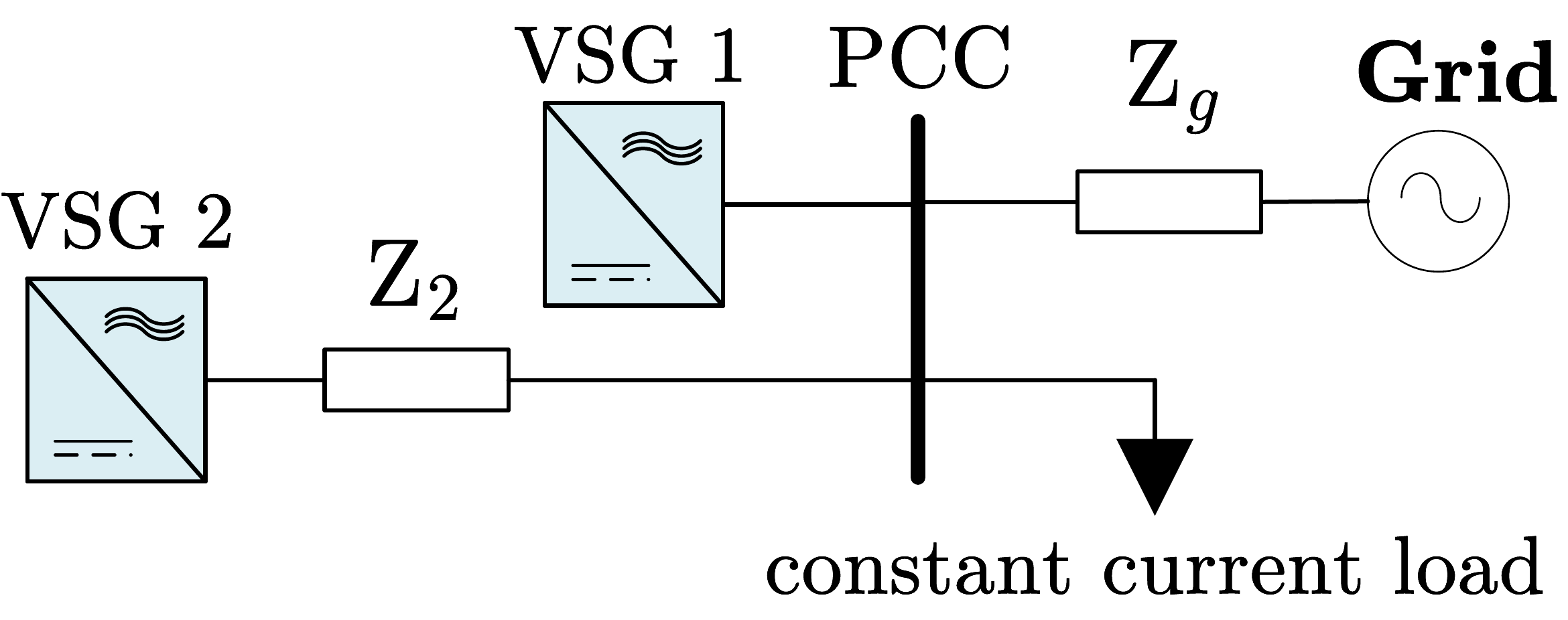}
  \caption{Grid-connected microgrid used in Application 1, with two VSGs and a local constant-current load.}
  \label{fig:case1_network}
\end{figure}

\begin{figure*}[!t]
  \centering
  \includegraphics[width=0.7\textwidth]{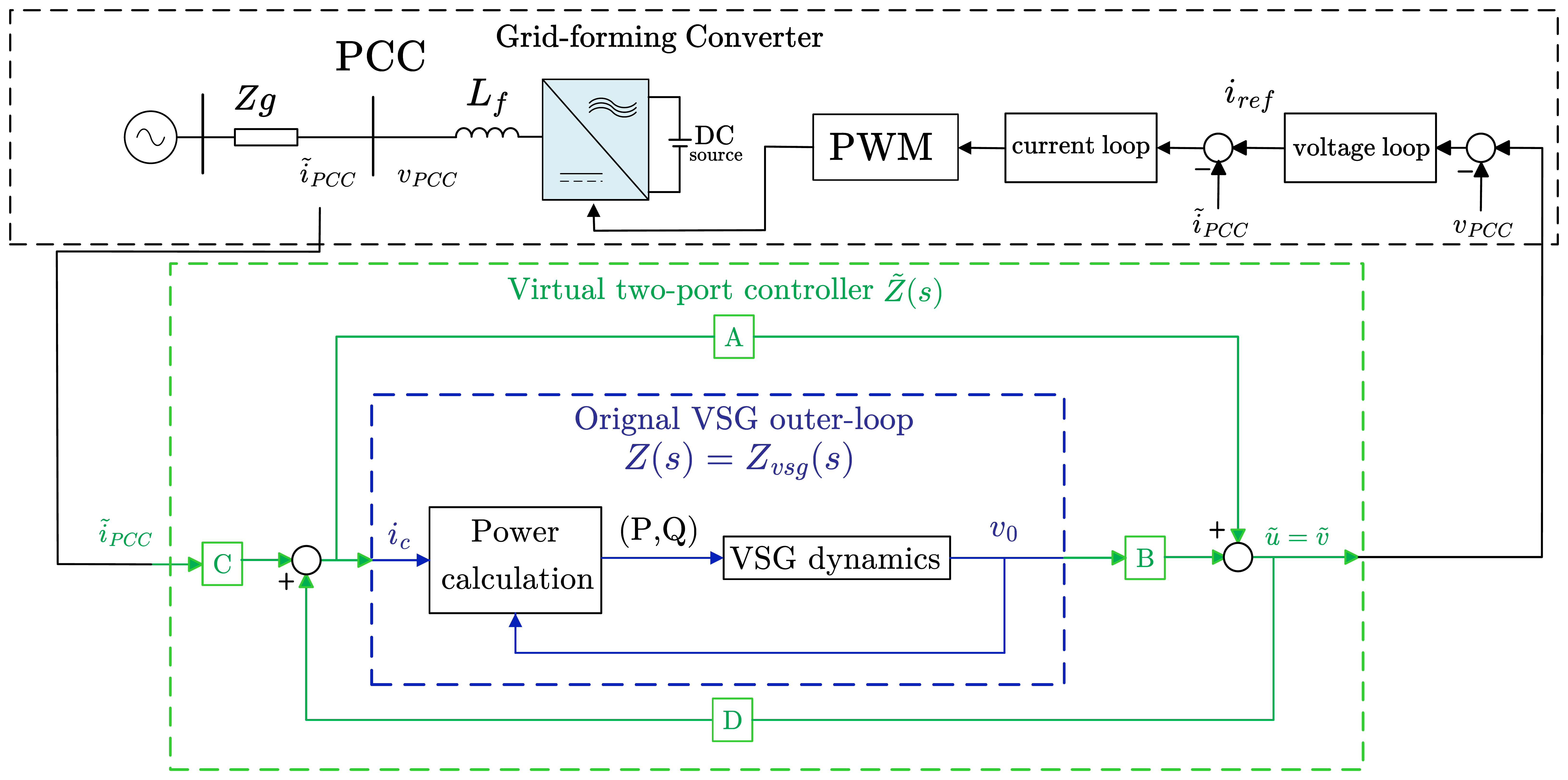}
  \caption{The virtual
  two-port control on a VSG.}
  \label{fig:case_common_model}
\end{figure*}

The VSG outer loop is linearized about the operating point and expressed in
the Laplace domain as
\begin{align*}
 (Js+D_\omega)\Delta\omega(s)
 &=\Delta P^\star(s)-\Delta P(s),\\
 s\Delta\theta(s)
 &=\Delta\omega(s),\\
 \Delta E_{\rm ref}(s)
 &=K_q\!\left[\Delta Q^\star(s)-\Delta Q(s)\right],\\
 \Delta v_0(s)
 &=e_{\theta0}\Delta E_{\rm ref}(s)
   +E_{{\rm ref},0}J_\theta e_{\theta0}\Delta\theta(s),
\end{align*}
where
\[
 e_{\theta0}
 =
 \begin{bmatrix}
 \cos\theta_0\\
 \sin\theta_0
 \end{bmatrix},
 \qquad
 J_\theta=
 \begin{bmatrix}
 0&-1\\
 1&0
 \end{bmatrix}.
\]
The incremental active and reactive powers are
\begin{align*}
 \Delta P(s)
 &=-i_{c0}^{T}\Delta v_0(s)-v_{00}^{T}\Delta i_c(s),\\
 \Delta Q(s)
 &= i_{c0}^{T}J_\theta\Delta v_0(s)
   -v_{00}^{T}J_\theta\Delta i_c(s).
\end{align*}
The inner loops with PI control are
\begin{equation*}
 \begin{aligned}
 \tilde v&=A i_c+Bv_0,&
 i_c&=C\tilde i_{\rm PCC}+D\tilde v,\\
 G_v(s)&=P_1+I_1/s,&G_i(s)&=P_2+I_2/s,\\
 \tilde i_{\rm PCC}^{\star}&=-G_v(\tilde v-v_{\rm PCC}),&
 L_f s\tilde i_{\rm PCC}
 &=G_i(\tilde i_{\rm PCC}^{\star}-\tilde i_{\rm PCC}).
 \end{aligned}
\end{equation*}

We retain \(f_\star=0.4\) Hz as a representative frequency with pronounced
VI saturation. The test compares command demand and its microgrid consequences
under the same pointwise target, \(\Omega=\{2\pi f_\star\}\).
Table~\ref{tab:case1_fixed_vsg} lists the common settings; VSG~2 uses the same
PI gains, filter, and limits, with its voltage/reactive-power references
initialized at the selected equilibrium.

\begin{table}[!htbp]
\centering
\caption{Application 1 operating condition and evaluation settings.}
\label{tab:case1_fixed_vsg}
\footnotesize
\setlength{\tabcolsep}{4pt}
\begin{tabular}{@{}ll@{}}
\toprule
Quantity & Value \\
\midrule
\((P_1,I_1)\), \((P_2,I_2)\) &
\((4.30,3.40\ {\rm s}^{-1})\), \((7.84,4.70\ {\rm s}^{-1})\) \\
\(L_f\), \(\omega_0\) &
\(3.54\times10^{-3}\) pu s, \(2\pi\times50\) rad/s \\
\(E_0\), \(K_q\), \(Q^\star\) & \(1.02\), \(0.011\), \(-0.12\) pu \\
\(P^\star\), \(J\), \(D_\omega\) &
\(0.45\) pu, \(0.48\), \(1.52\) \\
Required margin \(\varepsilon\) & \(2\) pu at \(0.4\) Hz \\
Current-reference and voltage limits & \(1.10\) pu, \(1.20\) pu \\
Feeder impedances \(Z_2=Z_g\) & \(0.188+\mathrm{j}0.516\) pu \\
VSG~2: \(P^\star,J,D_\omega\) & \(0.23417,1.0,0.30\) \\
Utility/PCC voltage magnitudes & \(0.96478/1.01872\) pu \\
VSG~2 terminal voltage magnitude & \(1.06556\) pu \\
Local load current \((d,q)\) & \((0.2204,0.0142)\) pu \\
\bottomrule
\end{tabular}
\end{table}

At this operating condition, the original pointwise passivity margin is $\mu_0= \lambda_{\min}\!\left( \operatorname{He}\{Z(\mathrm{j}2\pi f_\star)\}\right) =-6.79161$. We choose \(\varepsilon=2\) as an illustrative positive passivity target
for both controllers.
By Theorem~\ref{thm:finite_gain_vi_limit}, the minimum VI gain for the prescribed margin is $\gamma_{\rm VI}^{\star}=2-\mu_0=8.79161$.
For the two-port control, we use \(\rho=0.29\). Substitution of
\(\varepsilon=2\), \(\mu_0=-6.79161\), and \(\rho=0.29\) into
\eqref{eq:reduced_gain_tp_blocks} gives $A=3.969568I, B=0.29I, C=I, D=0.$
Direct evaluation gives 
\(\gamma_{\rm TP}=4.03256<\gamma_{\rm VI}^{\star}=8.79161\).

\par\addvspace{4pt}
\noindent\begin{minipage}{\columnwidth}
\makeatletter\def\@captype{figure}\makeatother
\centering
\includegraphics[width=0.92\columnwidth]
{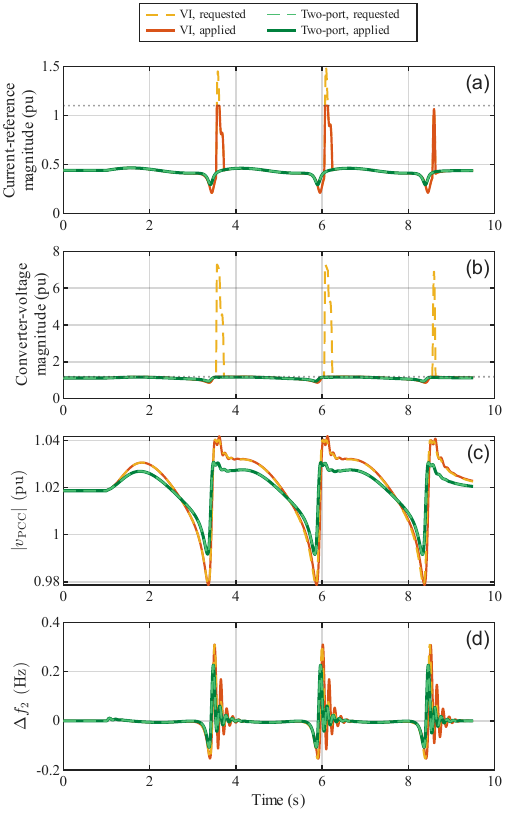}
\caption{Two-VSG microgrid test at \(0.4\) Hz. In (a), (b), dashed
and solid curves denote the requested and applied VSG~1 commands,
respectively, with dotted limit levels. In (c), (d), dashed and solid curves
denote the unlimited and limited PCC-voltage and VSG~2-frequency responses,
respectively. Orange: VI; green: two-port.}
\label{fig:case1_finite_gain_waveform}
\end{minipage}
\par\addvspace{4pt}

A \(0.03\)-pu, \(0.4\)-Hz sinusoidal measurement disturbance is applied
to the \(d\)-axis current reconstruction channel of VSG~1 from \(t=1\) to
\(8.5\) s. The network, operating point, disturbance, and hardware limits
are identical for VI and two-port control. For each strategy, a companion
simulation removes only VSG~1's current-reference and voltage limits, so the
paired responses isolate the effect of these limits.

Fig.~\ref{fig:case1_finite_gain_waveform}(a), (b) shows that VI requests peak
current-reference and converter-voltage magnitudes of \(1.48809\) and
\(7.28267\) pu. Their limits are active for \(1.67\%\) and \(5.09\%\)
of the disturbance interval. Two-port peaks are \(0.46658\) and
\(1.17846\) pu, so neither limit is activated. VSG~2 remains unsaturated
under both strategies. Panels (c), (d) show the corresponding
PCC-voltage and VSG~2-frequency responses. The limited and unlimited VI
trajectories separate by up to \(0.01684\) pu and \(0.16957\) Hz,
respectively, whereas the two-port trajectories coincide with their unlimited
counterparts. Moreover, under the same disturbance, two-port control produces
smaller system-level excursions: the minimum PCC voltage is \(0.99168\) pu,
compared with \(0.97874\) pu for VI, while the peak absolute VSG~2-frequency
deviation is reduced from \(0.30823\) Hz under VI to \(0.22592\) Hz.
Thus, VI saturation alters the controlled terminal response and propagates
through the microgrid to the neighboring VSG. By avoiding saturation, the
two-port implementation preserves its intended response and, under the tested
condition, produces smaller PCC-voltage and neighboring-VSG frequency
excursions than VI.

\FloatBarrier
\section{Application 2: Uncertainty Compression}
\label{sec:impedance_shaping_application}

\subsection{Design Objective}

In microgrid operation, variations in converter parameters and implementation
can cause the terminal characteristic presented to the network to deviate
from its nominal behavior \cite{Ying2025RobustPassive,Cecati2025Interoperability}.
The network therefore interacts not with one fixed terminal model, but
potentially with an uncertainty family of terminal impedances. An important question is whether one fixed terminal controller
can reduce how strongly such internal uncertainty is reflected in an
externally observed interaction characteristic.

To quantify how the uncertainty is reflected at the terminal, we use the
passivity-margin interval generated by the impedance family as a
representative scalar domain. Its width provides a stability-relevant measure
of the uncertainty-induced variation and its compression.

As the simplest representative case, we consider a scalar uncertain control
parameter \(p\) that varies continuously over the closed interval
\(\mathcal P=[p_-,p_+]\), and define
\begin{equation}
 \eta(\omega,p):=
 \lambda_{\min}\!\left[
 \operatorname{He}\{Z(\mathrm{j}\omega,p)\}
 \right].
 \label{eq:pointwise_passivity_index}
\end{equation}
The quantity \(\eta(\omega,p)\) is the pointwise passivity margin; a
nonnegative value means that the member indexed by \(p\) is passive at
frequency \(\omega\).
At each frequency, the parameter interval produces
\begin{equation}
 \mathcal I_\mu(\omega)
 :=
 \left[
 \underline\eta(\omega),\overline\eta(\omega)
 \right]
 =
 \left[
 \inf_{p\in\mathcal P}\eta(\omega,p),
 \sup_{p\in\mathcal P}\eta(\omega,p)
 \right].
 \label{eq:passivity_margin_interval}
\end{equation}
For each frequency, the uncertain impedance family is mapped through
\(\eta(\omega,p)\) into the interval \(\mathcal I_\mu(\omega)\). The width of
this interval therefore measures how strongly the underlying parameter
uncertainty is reflected in the selected terminal-interaction measure. To
quantify this uncertainty-induced spread over \(\Omega\), define
\begin{equation}
 \Delta_\mu(\mathcal Z;\Omega)
 :=
 \sup_{\omega\in\Omega}
 \left[
 \overline\eta(\omega)-\underline\eta(\omega)
 \right].
 \label{eq:passivity_margin_spread}
\end{equation}
Here, \(\Delta_\mu\) is used primarily as a measure of the uncertainty-induced
spread in the selected terminal characteristic. A smaller value means that
the observable is less sensitive to the underlying parameter uncertainty.

\subsection{VI Control: Translation Without Compression}

For VI control, the same \(Z_v\) is added to every member of the
impedance family, so all pairwise impedance differences are preserved. This directly leads to the following theorem on this additive translation property.

\begin{theorem}[VI translation and uncertainty-domain invariance]
\label{thm:vi_family_translation}
Let \(Z(\cdot,p)\in\mathcal{RH}_{\infty}^{2\times2}\) for
\(p\in\mathcal P\), and let \(Z_v\in\mathcal{RH}_{\infty}^{2\times2}\) be
independent of \(p\). For a desired lower bound \(m>0\), if
\begin{equation}
  \operatorname{He}\{Z_v(\mathrm{j}\omega)\}
  \succeq
  [m-\underline\eta(\omega)]I,
  \qquad \omega\in\Omega,
  \label{eq:general_vi_passivation_condition}
\end{equation}
then
\(
  \eta_{\rm VI}(\omega,p)\geq m.
\)
Moreover, the additive map preserves every pairwise impedance difference:
\begin{equation}
 \widetilde Z_{\rm VI}(\cdot,p_1)-
 \widetilde Z_{\rm VI}(\cdot,p_2)
 =Z(\cdot,p_1)-Z(\cdot,p_2).
 \label{eq:general_vi_pairwise_invariance}
\end{equation}
If, in addition, the VI is isotropic, i.e.,
\begin{equation}
  \operatorname{He}\{Z_v(\mathrm{j}\omega)\}
  =r_v(\omega)I,
  \qquad r_v(\omega)\in\mathbb R,
  \quad \omega\in\Omega,
  \label{eq:isotropic_vi_definition}
\end{equation}
then every family member receives the same scalar margin shift and $\Delta_\mu(\widetilde{\mathcal Z}_{\rm VI};\Omega) = \Delta_\mu(\mathcal Z;\Omega)$.
\end{theorem}

\noindent\emph{Proof:} See Appendix~\ref{app:proof_vi_family_translation}.

This theorem reveals the VI-side limitation in uncertainty management. A fixed
VI can translate the entire family, but it preserves all pairwise impedance
differences. Under isotropic VI, this translation also leaves
\(\Delta_\mu\) unchanged. Hence VI can change the location of the selected
uncertainty domain but cannot reduce its spread.

\subsection{Two-Port Design: Uncertainty-Domain Compression}

We next ask whether the virtual two-port can reduce the spread of the same
uncertainty-induced domain rather than merely translate it. Its
linear-fractional map provides this additional capability by allowing both the
location and the diameter of the mapped domain to be reshaped, as stated by
the following theorem.

\begin{theorem}[Uncertainty-domain compression by virtual two-port control]
\label{thm:cayley_robust_passivation}
Let \(Z(\cdot,p)\in\mathcal{RH}_{\infty}^{2\times2}\) for
\(p\in\mathcal P\), and suppose that for some \(p_c\in\mathcal P\),
\[
 \sup_{p\in\mathcal P}
 \|Z(\cdot,p)-Z(\cdot,p_c)\|_\infty<\infty.
\]
For every lower bound \(m>0\) and diameter bound
\(\overline\Delta>0\), there exists a two-port control such that the
reconstructed terminal relation is stable and proper and, for all
\(p\in\mathcal P\) and \(\omega\in\Omega\), satisfies $ m \leq \eta_{\rm TP}(\omega,p) \leq m+\overline\Delta$,
and therefore $\Delta_\mu(\widetilde{\mathcal Z}_{\rm TP};\Omega) \leq \overline\Delta$.
\end{theorem}

\noindent\emph{Proof:} See Appendix~\ref{app:proof_cayley_robust_passivation}.

Theorem~\ref{thm:cayley_robust_passivation} separates the location and width
of the uncertainty-induced domain. The parameter \(m\) sets its lower
location, whereas \(\overline\Delta\) bounds its dispersion across the
uncertain family. In particular, choosing \(m>0\) places the entire compressed
domain in the passive region, while \(\overline\Delta\) can be selected
independently to limit its width.

The structural distinction from VI is therefore not whether the uncertain
family can simply be shifted to a desired region, but whether its
uncertainty-induced spread can be reduced. Fixed isotropic VI preserves this
spread, whereas the virtual two-port can make it arbitrarily small under the
conditions of the theorem. In this sense, the two-port provides an additional
means of managing how strongly internal converter uncertainty is manifested
in the selected network-facing interaction characteristic.

\subsection{Case Study}
\label{subsec:case2_continuous_compression}

The same VSG structure as in Application~1 is now evaluated over a continuous
parameter family to evaluate the uncertain terminal characteristic that
a converter presents to the microgrid. This local frequency-domain test
quantifies compression of that characteristic. We set \(P^\star=0.45\) pu, \(J=0.6\), and the nominal
damping coefficient \(D_{\omega0}=1.9\), while the actual coefficient varies
according to
\begin{equation}
 p:=\frac{D_\omega}{D_{\omega0}}\in[0.8,1.2].
 \label{eq:case2_continuous_parameter}
\end{equation}
The nominal member \(p_c=1\) defines the reference impedance
\(Z_c:=Z(\cdot,p_c)\), and the interaction band is
\(f\in[0.01,0.5]\) Hz. The fixed VI and fixed two-port are then applied to
every member of this family.

Over the considered frequency band, the original family has a
passivity-margin spread of \(\Delta_\mu=67.436\) pu. For VI control, we use \(Z_v=237.171I\) to ensure positive passivity
margins for the entire uncertainty family over the considered frequency band. The fixed VI raises the minimum passivity margin to \(0.020\) pu
while leaving \(\Delta_\mu\) unchanged at \(67.436\) pu.

For the two-port control, we select \(m=4/3\) and
\(\overline\Delta=5/3\). Applying the proposed design yields
\begin{equation}
 \begin{aligned}
 A&=0.546835I-1.01295\times10^{-3}Z_c,\\
 B&=1.01295\times10^{-3}I,\\
 C&=1.828704I,\qquad D=0.914352I.
 \end{aligned}
 \label{eq:case2_final_abcd}
\end{equation}
The two-port achieves a minimum passivity margin of \(1.679\) pu and reduces
\(\Delta_\mu\) to \(0.321\) pu, below the prescribed bound
\(\overline\Delta=1.667\) pu.

It is also informative to evaluate the controller coefficients using the same
gain measure introduced in Application~1. For VI,
\begin{equation}
 \gamma_{\rm VI}
 =\|Z_v\|_\infty
 =237.17105,
 \label{eq:case2_vi_gain}
\end{equation}
whereas the two-port realization gives
\begin{equation}
 \gamma_{\rm TP}
 =
 \left\|
 \begin{bmatrix}
 A&B-I\\
 C-I&D
 \end{bmatrix}
 \right\|_\infty
 =1.71344.
 \label{eq:case2_tp_gain}
\end{equation}
Thus, the two-port combines substantial uncertainty-domain compression with a
considerably smaller coefficient gain than VI.

\begin{figure}[!b]
\centering
\includegraphics[width=0.80\columnwidth]
{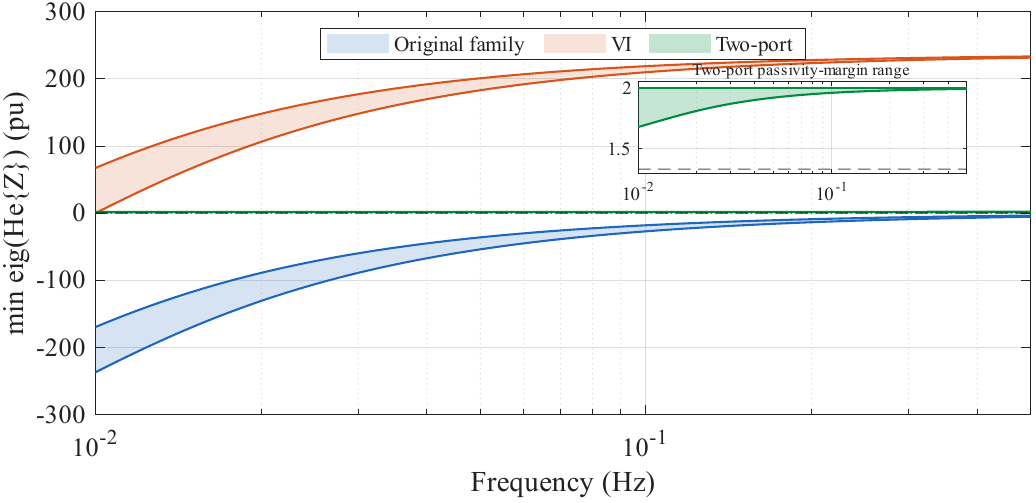}
\caption{Uncertainty-compression results for
\(D_\omega/D_{\omega0}\in[0.8,1.2]\). The fixed isotropic VI translates the
uncertainty-induced margin domain without changing its diameter, whereas the
virtual two-port substantially compresses the same domain while retaining a
positive lower margin.}
\label{fig:application2_continuous_compression}
\end{figure}

As shown in Fig.~\ref{fig:application2_continuous_compression}, the two
controllers treat the same uncertainty family in fundamentally different
ways. VI translates the selected terminal-characteristic domain while
preserving its uncertainty-induced spread, whereas the virtual two-port
reduces the diameter from \(67.436\) pu to \(0.321\) pu. Thus, the terminal
interaction measure becomes substantially less sensitive to the underlying
converter uncertainty. Together with the much smaller coefficient gain, this
result shows how local reconstruction reduces uncertainty in the
converter characteristic used for microgrid interaction analysis.

To assess how this local uncertainty compression propagates to network-level responses, Appendix~\ref{app:ieee33_network_validation}-A further evaluates the proposed mechanism on a modified IEEE 33-bus microgrid with eight VSGs.

\section{Application 3: Flexible Power-Flow Regulation}
\label{subsec:nonnegative_resistance_capability}

\subsection{Design Objective}

Beyond the dynamic terminal shaping in Applications~1 and~2, we now use
the same framework for steady-state coordination. In microgrids with high
\(R/X\) ratios or unequal feeder
impedances, identical converters can exhibit inaccurate power sharing. VI is
widely used to compensate this problem by reshaping the effective series
impedance. However, stronger
sharing correction generally requires a larger series voltage drop, which
limits the flexibility of VI-based power-flow regulation
. This application investigates whether the additional
degree of freedom in virtual two-port control can overcome this limitation.

\subsection{VI Control: Series-Only Compensation}

To compare the regulation capabilities, consider the common single-source
circuit in Fig.~\ref{fig:positive_pi_network}. The general reciprocal
\(\Pi\)-type realization contains a converter-side shunt admittance \(y_1\),
a series admittance \(y_2\), and a line-side shunt admittance \(y_3\).

\begin{figure}[!htbp]
  \centering
  \includegraphics[width=0.82\columnwidth]{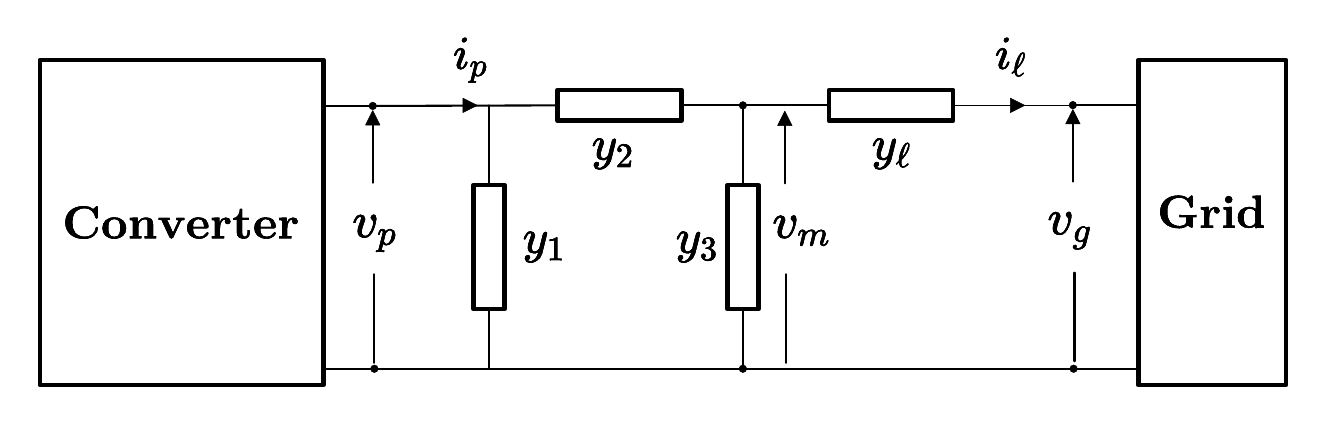}
  \caption{Reciprocal \(\Pi\)-type terminal realization connected to the
  physical line.}
  \label{fig:positive_pi_network}
\end{figure}

Here, \(v_p\), \(v_m\), and \(v_g\) are the converter-port,
controlled-terminal, and grid voltages, respectively, and \(y_\ell\) is the
physical-line admittance. For fixed \(v_g\) and \(y_\ell\), each \(v_m\)
determines the line current \(i_\ell=y_\ell(v_m-v_g)\) and hence a
power-flow operating point. The attainable set of \(v_m\) therefore
quantifies the power-flow regulation capability.

For VI, the two shunt branches are absent, i.e., \(y_1=y_3=0\), and the only
compensating element is the series impedance \(z_v=y_2^{-1}\). It can change
\(v_m\) only through a series voltage drop. For a
prescribed \(v_m\neq v_g\), the circuit equation uniquely determines the
required finite impedance as
\begin{equation}
 z_v=z_\ell\frac{v_p-v_m}{v_m-v_g},\qquad z_\ell=y_\ell^{-1}.
 \label{eq:vi_required_impedance}
\end{equation}
Its controlled-voltage set is
\begin{equation}
 \mathcal V_{\rm VI}
 =\left\{v_m\in\mathbb C\setminus\{v_g\}\;\middle|\;
 \operatorname{Re}\!\left[
 z_\ell(v_p-v_m)(v_m-v_g)^*
 \right]\geq0\right\}.
 \label{eq:vi_vm_half_plane_condition}
\end{equation}
The point \(v_m=v_p\) is included with \(z_v=0\). Hence
\eqref{eq:vi_vm_half_plane_condition} is not a tuning rule but a reachability
condition: points that violate it would require a negative-resistance series
branch. The excluded region is the structural benchmark for the two-port.

\subsection{Two-Port Control: Series-Shunt Compensation}

For the two-port considered here, both the series admittance \(y_2\) and the
line-side shunt admittance \(y_3\) in
Fig.~\ref{fig:positive_pi_network} participate in power-flow regulation. The
series branch provides series compensation through the voltage difference
\(v_m-v_p\), while the shunt branch provides parallel compensation directly
at \(v_m\). Their coordinated action extends the regulation capability beyond
that of series-only VI.

 Under the reciprocal-network realization, the controlled-terminal voltage $v_m$ satisfies
\begin{equation}
 y_2(v_m-v_p)+y_3v_m+y_\ell(v_m-v_g)=0.
 \label{eq:positive_pi_node_equation}
\end{equation}
For a prescribed $v_m$, \eqref{eq:positive_pi_node_equation} is an inverse
feasibility problem in $y_2$ and $y_3$. Two cases need to be distinguished.
If $v_m$ is not collinear with $v_p$, the two adjustable admittances provide
sufficient freedom to satisfy \eqref{eq:positive_pi_node_equation} with
positive conductances. If $v_m$ is collinear with $v_p$, additional
conditions are required. The following theorem gives the exact feasibility
conditions for both cases.

\begin{theorem}[Positive-resistance series-shunt voltage capability]
\label{thm:nonnegative_resistance_vm_reachability}
Assume $v_p\neq0$, $v_g\neq0$, $v_p\neq v_g$, and $y_\ell\neq0$. Define
\begin{equation*}
 w=\frac{v_m}{v_p},\qquad
 \eta_g=\frac{v_g}{v_p},\qquad
 \phi(t)=\operatorname{Re}\{y_\ell(\eta_g-t)\}.
\end{equation*}
With finite $y_2,y_3\in\mathbb C$ such that
$\operatorname{Re}(y_2),\operatorname{Re}(y_3)>0$
and $y_2+y_3+y_\ell\neq 0$, the normalized controlled-voltage set is
\begin{equation}
 \begin{aligned}
 \mathcal W_\Pi^+=
 &(\mathbb C\setminus\mathbb R)\cup(0,1)\\
 &\cup\{t\leq0:\phi(t)<0\}
 \cup\{t\geq1:\phi(t)>0\}.
 \end{aligned}
 \label{eq:vm_nonnegative_resistance_characterization}
\end{equation}
Hence every $v_m$ away from the line through $0$ and $v_p$ is reachable;
the sign conditions in
\eqref{eq:vm_nonnegative_resistance_characterization} characterize the
remaining collinear points.
\end{theorem}

\noindent\emph{Proof:} See Appendix~\ref{app:proof_nonnegative_resistance_vm_reachability}.

Theorem~6 shows why the series-shunt realization has a fundamentally larger
voltage-reachability set than series VI. The shunt branch introduces the
additional term $y_3v_m$, so the required network-current balance is no longer
carried by the series branch alone. For noncollinear $v_m$, the two adjustable
branch terms provide sufficient independent freedom to satisfy the node
equation. Only when $v_m$ is collinear with $v_p$ does this freedom degenerate,
leaving the one-dimensional feasibility restrictions in \eqref{eq:vm_nonnegative_resistance_characterization}. The converter terminal can therefore participate in the network current
balance through both series- and shunt-like actions, rather than relying
solely on a larger series voltage drop.

\begin{figure}[!htbp]
  \centering
  \includegraphics[width=0.86\columnwidth]
  {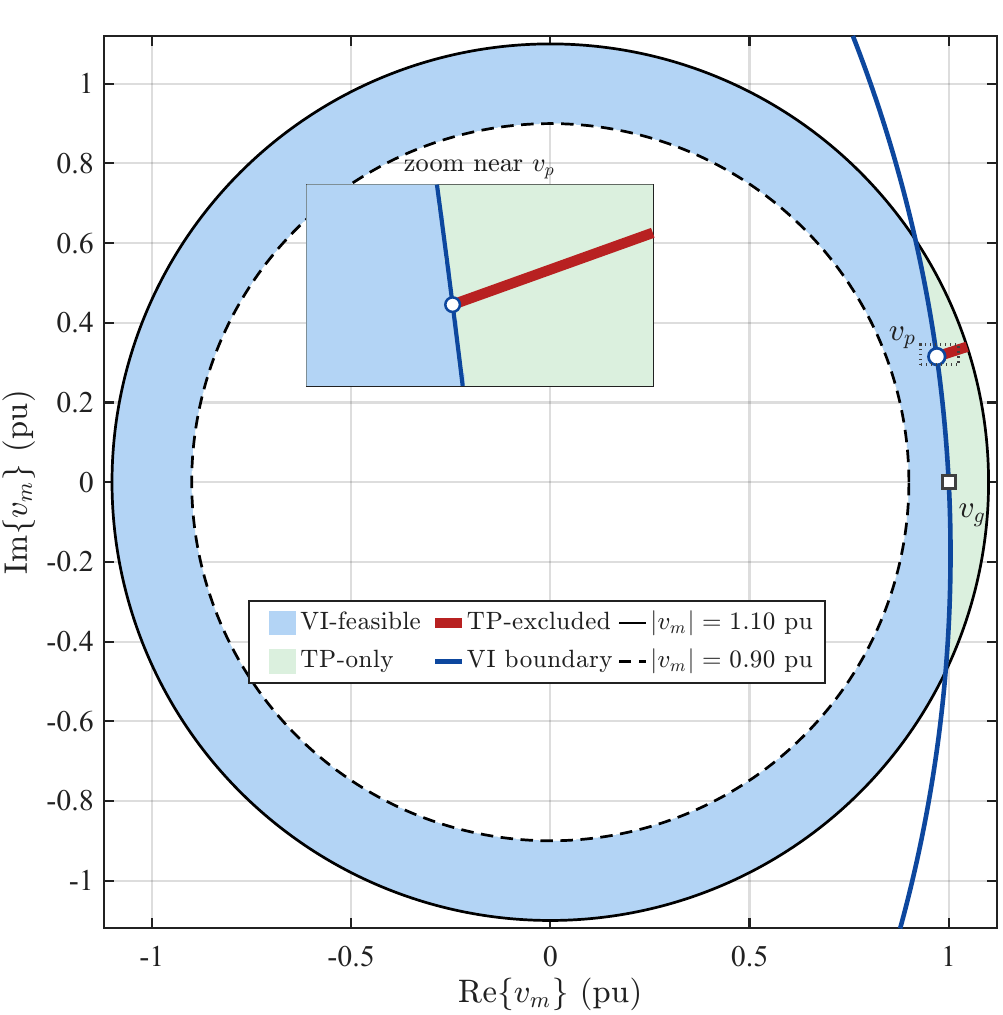}
  \caption{Algebraically reachable voltage sets within
$0.90\leq|v_m|\leq1.10$ pu for positive-resistance series-shunt control
and nonnegative-resistance series VI.}
  \label{fig:vm_feasible_region_comparison}
\end{figure}

The reachability difference predicted by
Theorem~\ref{thm:nonnegative_resistance_vm_reachability}
is illustrated in Fig.~\ref{fig:vm_feasible_region_comparison} for
$v_p=1.02\angle18^\circ$ pu, $v_g=1\angle0^\circ$ pu, and
$z_\ell=0.02+\mathrm{j}0.45$ pu.
The blue region is reachable by series VI according to
\eqref{eq:vi_vm_half_plane_condition}, while the green region is enabled by
the additional series-shunt degree of freedom. The remaining red radial
segment corresponds to the collinear restriction identified in
Theorem~\ref{thm:nonnegative_resistance_vm_reachability}.
Within the considered voltage annulus, series VI reaches $94.03\%$ of the
evaluated points, whereas the positive-resistance series-shunt realization
reaches all points except this one-dimensional segment. Thus, the added shunt
degree of freedom reduces the VI-infeasible region from a finite-area subset
to a one-dimensional exceptional set while retaining positive branch
resistance. This comparison characterizes algebraic voltage reachability under the
stated branch-resistance conditions. Practical implementation further
requires admissible converter voltage and current commands, modulation
headroom, and stable operation within the available control bandwidth.
These constraints restrict the usable subset of the plotted region.
The theoretical result identifies the additional reachability provided by
series--shunt compensation; the achievable improvement under implementation
limits depends on the operating conditions.

\subsection{Case Study}
\label{subsec:case2_upfc}

We next evaluate this capability in the asymmetric ring microgrid of
Fig.~\ref{fig:case2_system}. Two identical VSGs are connected to the load bus
through unequal transfer paths. Although their converter parameters and power
references are identical, the unequal paths give the two source ports
different self and source-to-load transfer terms. Equalizing these terms makes
the external network invariant to interchange of the two source ports and
removes the transfer-path asymmetry seen by the identical converters. This is
the mirror-symmetry objective used below. Independent series VI can modify the
two self impedances; the series-shunt two-port can also modify the source-to-
load transfer terms. The controller is designed from this network objective,
after which the load-bus voltage and power-sharing responses are evaluated
under the load sequence.
\begin{figure}[!htbp]
  \centering
  \includegraphics[width=0.76\columnwidth]{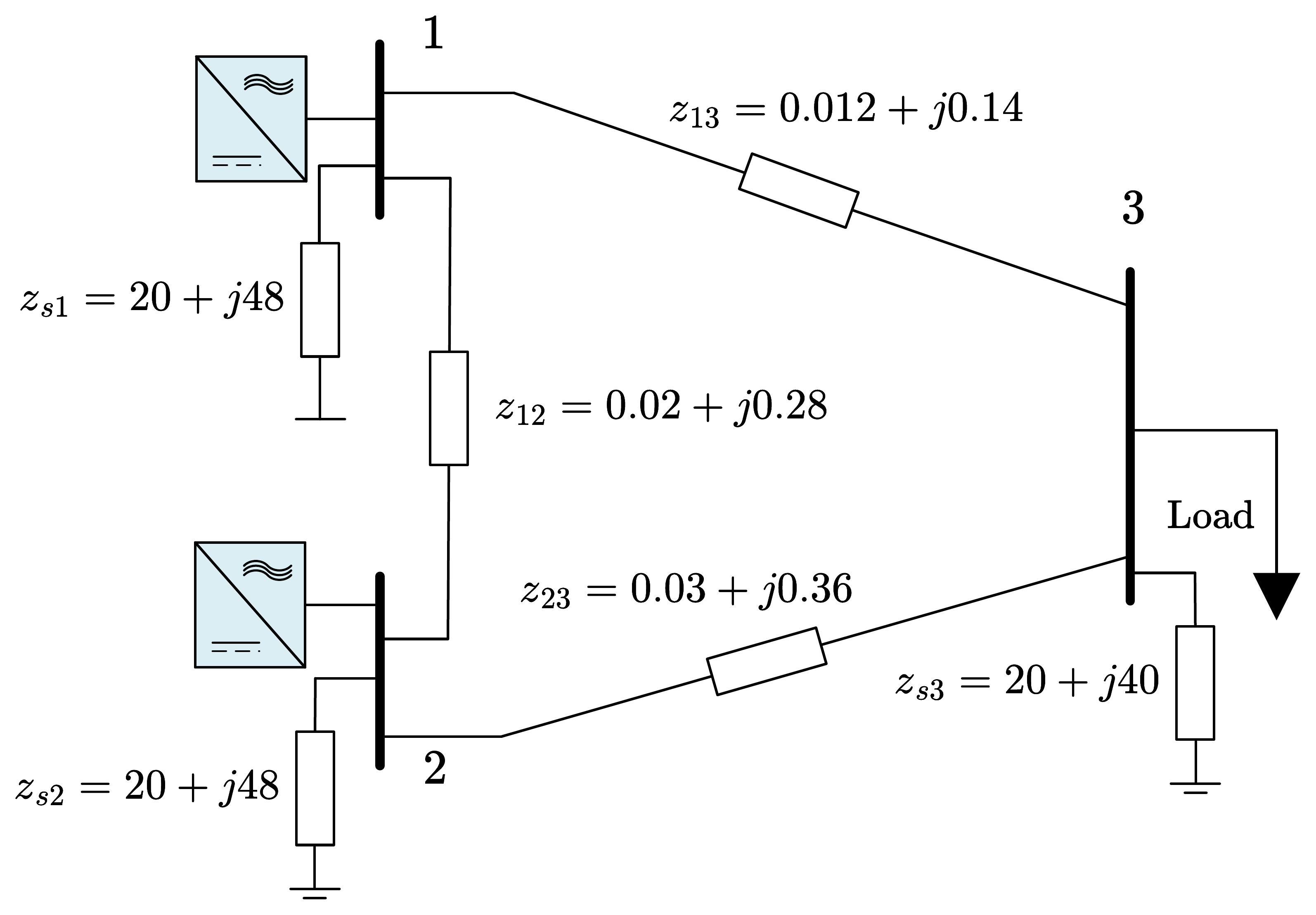}
  \caption{Asymmetric ring microgrid supplied by two identical VSGs.}
  \label{fig:case2_system}
\end{figure}

\begin{table}[!htbp]
  \centering
  \caption{Application 3 system and controller settings.}
  \label{tab:case2_settings}
  \footnotesize
  \setlength{\tabcolsep}{3.0pt}
  \begin{tabular}{@{}p{0.34\columnwidth}p{0.58\columnwidth}@{}}
    \toprule
    Quantity & Value \\
    \midrule
    Line impedances & $z_{12}=0.020+\mathrm{j}0.280$,
      $z_{13}=0.012+\mathrm{j}0.140$,
      $z_{23}=0.030+\mathrm{j}0.360$ pu \\
    Local shunts & $z_{s1}=z_{s2}=20+\mathrm{j}48$,
      $z_{s3}=20+\mathrm{j}40$ pu \\
    VSG parameters & $E_0=1.04$, $J=8.0\times10^{-4}$,
      $D_\omega=1.70$, $K_q=0.03$ \\
    Power references & $P^\star=0.43361$, $Q^\star=0.17698$
      for both VSGs \\
    Filter and PI gains & $L_f=2.0\times10^{-3}$ pu s,
      $(P_1,I_1)=(5,2.5\ {\rm s}^{-1})$,
      $(P_2,I_2)=(5.3366,2.5\ {\rm s}^{-1})$ \\
    Load sequence & $0.82+\mathrm{j}0.22$, $1.00+\mathrm{j}0.30$,
      $0.90+\mathrm{j}0.26$, $1.08+\mathrm{j}0.34$ pu \\
    Switching instants & $t=10,20,30$ s; $t_{\rm end}=40$ s \\
    Operating-voltage range & $0.95\leq|v_3|\leq1.05$ pu \\
    \bottomrule
  \end{tabular}
\end{table}

After eliminating the physical PCC nodes, let the retained order of the
three-port admittance $Y_{\rm red}$ be the two reconstructed source ports and
the load bus. The first two equalities below match the source self and
source-to-load terms; reciprocity gives the corresponding load-to-source
condition. Mirror symmetry therefore requires
\begin{equation}
 Y_{11}=Y_{22},\qquad Y_{13}=Y_{23},\qquad Y_{31}=Y_{32},
 \label{eq:case2_mirror_conditions}
\end{equation}
and is measured by
\begin{equation}
 \varepsilon_{\rm sym}=
 \frac{\sqrt{|Y_{11}-Y_{22}|^2+|Y_{13}-Y_{23}|^2+|Y_{31}-Y_{32}|^2}}
 {\|Y_{\rm red}\|_{\rm F}}.
 \label{eq:case2_symmetry_error}
\end{equation}
The VI benchmark minimizes this same index subject to the operating-voltage
range in Table~\ref{tab:case2_settings}. In the three-port impedance representation,
its structural action is
\begin{equation}
 Z_{\rm VI}=Z_{\rm red}^{\rm base}
 +\operatorname{diag}\{z_{v,1},z_{v,2},0\},
 \label{eq:case2_vi_impedance_update}
\end{equation}
where $Z_{\rm red}^{\rm base}$ is the reduced three-port impedance of the
physical network before terminal compensation. Series VI changes the two
source self terms while retaining the original transfer terms. The constrained
minimum is obtained with
$z_{v,1}=0.012010+\mathrm{j}0.024624$ pu and
$z_{v,2}=0.010169+\mathrm{j}0.000044$ pu. For the present network,
$Z_{13}^{\rm base}-Z_{23}^{\rm base}
=0.002494+\mathrm{j}0.024850$ pu, which leaves a finite mirror residual under every pair of series branches. 

The nonzero transfer-term mismatch therefore prevents series VI from
achieving exact mirror symmetry. In contrast, the series-shunt realization provides an additional degree of freedom to modify the transfer terms as well as the source self terms. We use the open converter-side shunt specialization in \eqref{eq:open_converter_shunt_abcd} and determine its coefficients by
minimizing \(\varepsilon_{\rm sym}\) subject to the operating-voltage
constraint in Table~\ref{tab:case2_settings}. The resulting coefficients are
\begin{equation}
\begin{aligned}
 A_1&=0.006856637600+\mathrm{j}0.039440081202,\\
 B_1&=C_1=1,\\
 D_1&=-0.011926722623+\mathrm{j}0.003561361920,\\
 A_2&=0.000968346384-\mathrm{j}0.039440081685,\\
 B_2&=C_2=1,\\
 D_2&=-0.000100000000-\mathrm{j}0.009443750550.
\end{aligned}
\label{eq:case2_abcd_values}
\end{equation}

The power-flow consequence of the network reshaping is evaluated at the
physical PCC, where
\begin{equation}
 S_{{\rm PCC},k}
 =-v_{{\rm PCC},k}\tilde i_{{\rm PCC},k}^{*}
 =P_{{\rm PCC},k}+\mathrm{j}Q_{{\rm PCC},k}.
 \label{eq:case2_pcc_power}
\end{equation}
The corresponding power increments are referenced to the first load stage.
\begin{figure}[!htbp]
  \centering
  \includegraphics[width=0.86\columnwidth]
  {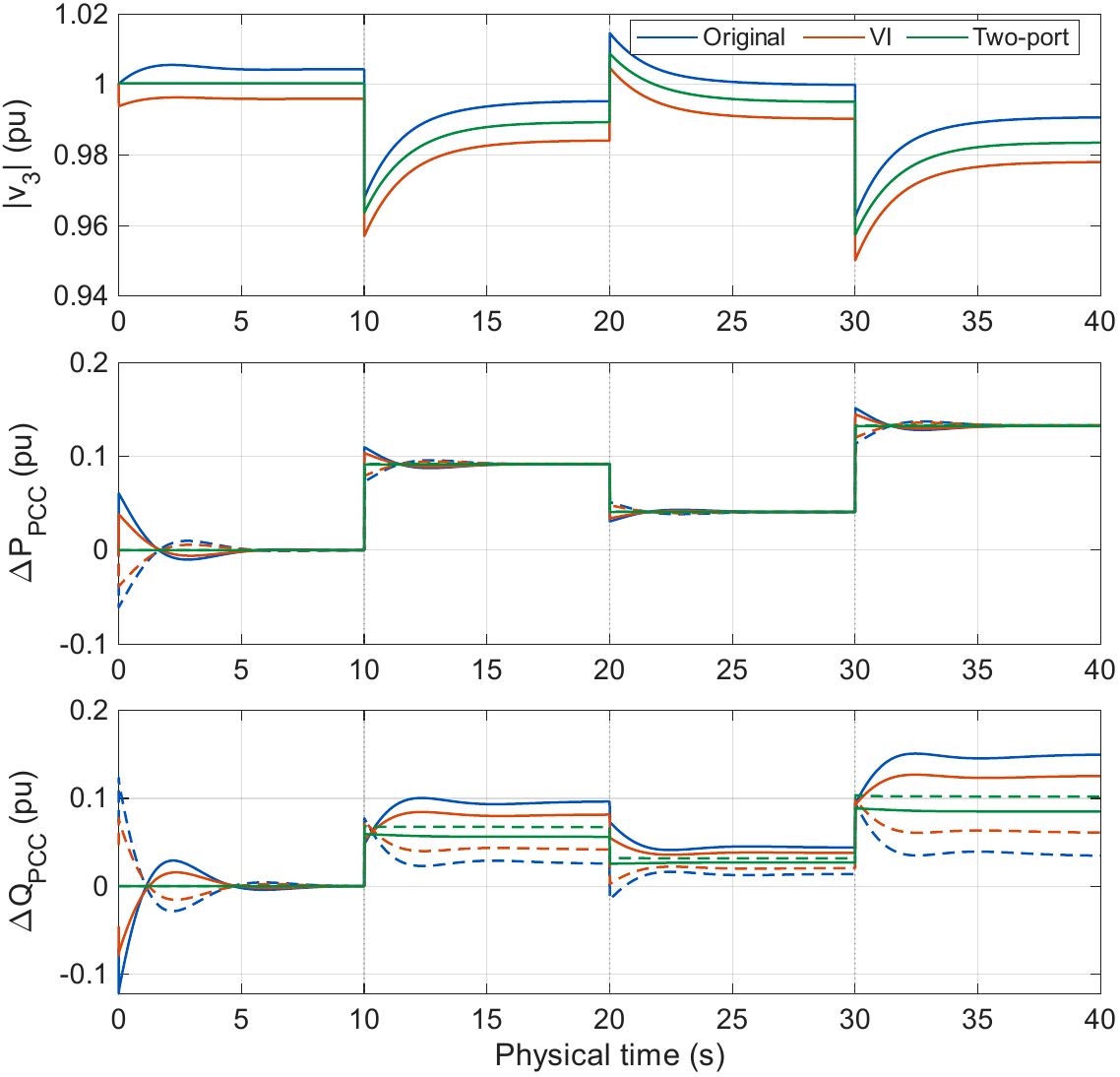}
\caption{Application 3 power-flow responses: common-bus voltage and
physical-PCC active- and reactive-power increments under the four-stage load
sequence. Solid and dashed curves denote VSGs 1 and 2, respectively.}
  \label{fig:case2_external_port_delta_pq}
\end{figure}
\begin{table}[!htbp]
  \centering
  \caption{Application 3 network-symmetry and power-flow results.}
  \label{tab:case2_symmetry_consequences}
  \scriptsize
  \setlength{\tabcolsep}{3.0pt}
  \resizebox{\columnwidth}{!}{%
  \begin{tabular}{lccc}
    \toprule
    Metric & Original & VI & Two-port\\
    \midrule
    $\varepsilon_{\rm sym}$
      & $3.79\times10^{-1}$ & $2.39\times10^{-1}$ & $1.88\times10^{-16}$\\
    $\max|\Delta P_{\rm PCC,1}-\Delta P_{\rm PCC,2}|$
      & $6.70\times10^{-4}$ & $1.24\times10^{-3}$ & $6.57\times10^{-4}$\\
    $\max|\Delta Q_{\rm PCC,1}-\Delta Q_{\rm PCC,2}|$
      & $1.15\times10^{-1}$ & $6.45\times10^{-2}$ & $1.72\times10^{-2}$\\
    $|v_3|$ range
      & 0.962-1.015 & 0.950-1.005 & 0.957-1.009\\
    \bottomrule
  \end{tabular}}
\end{table}

Fig.~\ref{fig:case2_external_port_delta_pq} shows the PCC power responses
throughout the four-stage load sequence, while
Table~\ref{tab:case2_symmetry_consequences} summarizes the corresponding
network-symmetry, power-mismatch, and voltage metrics. VI reduces
$\varepsilon_{\rm sym}$ from $0.379$ to $0.239$, whereas the series-shunt
two-port nearly eliminates the network asymmetry, reducing it to
$1.88\times10^{-16}$. The improvement is reflected most clearly in the
reactive-power response: the maximum PCC reactive-power increment mismatch
decreases from $0.115$ pu in the original network and $0.0645$ pu with VI to
$0.0172$ pu with the two-port. The active-power mismatches remain small in all
three cases. Meanwhile, the minimum load-bus voltage is $0.957$ pu with the
two-port, compared with $0.950$ pu with VI.

The results reveal the practical value of the added shunt degree of freedom.
Series VI modifies only the source self terms, whereas the two-port can also
reshape the source-to-load transfer terms. This additional structural freedom
allows the two-port to remove the transfer-path asymmetry much more completely,
which is reflected in the substantially closer reactive-power responses at the
physical PCCs. At the same time, the two-port produces a smaller load-bus
voltage drop than VI in the tested operating range. The scalability of this series–shunt coordination mechanism is further examined in Appendix~\ref{app:ieee33_network_validation}-B on a modified IEEE 33-bus microgrid with eight VSGs under multiple loading and feeder-reconfiguration conditions.

\vspace{-3pt}
\section{Conclusion}
\label{sec:conclusion}

This paper has developed virtual two-port control as a structured way to
reshape the terminal characteristics governing converter-network
interactions in inverter-based microgrids. Its four control channels jointly reconstruct terminal
voltage and current, producing a linear-fractional impedance map with VI as a
degenerate case. The derived return-inverse criterion guarantees stability and
properness of the controlled impedance.
The theoretical analyses and case studies demonstrate three advantages
over VI: a prescribed passivity margin can be achieved with lower controller
coefficient gain, avoiding limit-induced microgrid response deviations in
the tested case; uncertainty-induced variation in a network-facing terminal
characteristic can be substantially compressed rather than merely translated;
and power sharing can be improved with less voltage drop through coordinated
series-shunt compensation. These results show that reconstructing both voltage and current provides
additional structural freedom for dynamic interaction shaping and steady-state
power-flow regulation in inverter-based microgrids.

\appendices

\setlength{\abovedisplayskip}{1pt plus 0.3pt minus 0.3pt}
\setlength{\belowdisplayskip}{1pt plus 0.3pt minus 0.3pt}
\setlength{\abovedisplayshortskip}{0.5pt plus 0.3pt minus 0.3pt}
\setlength{\belowdisplayshortskip}{0.5pt plus 0.3pt minus 0.3pt}

\section{Proof of Theorem~\ref{thm:stable_reconstruction}}
\label{app:proof_stable_reconstruction}

Because \(H,D\in\mathcal{RH}_{\infty}^{2\times2}\), the matrix
\(I-DH\) has no poles in the closed right half-plane. Its inverse is proper
if and only if \(I-D(\infty)H(\infty)\) is nonsingular, and it is stable if
and only if \(\det[I-D(s)H(s)]\) has no zeros for
\(\operatorname{Re}(s)\geq0\). These two requirements give
the conditions in Theorem~\ref{thm:stable_reconstruction}. When they hold, closure of
\(\mathcal{RH}_\infty\) under products gives
\(SC,HSC\in\mathcal{RH}_{\infty}^{2\times2}\). Therefore,
\[
 i=SC\tilde i,\qquad
 \tilde u=HSC\tilde i,
 \qquad
 \widetilde Z=HSC\in\mathcal{RH}_{\infty}^{2\times2}.
\]
\section{Finite Reciprocal $\Pi$-Equivalent Conditions}
\label{app:reciprocal_pi_details}
We first introduce the notion of symmetric impedance. A symmetric impedance satisfies
\(Z_{dd}(s)=Z_{qq}(s)\) and \(Z_{qd}(s)=-Z_{dq}(s)\). Accordingly, define
\begin{equation*}
\mathcal S_{dq}
:=
\left\{
\begin{bmatrix}
a&-b\\
b&a
\end{bmatrix}
\;\middle|\;
a,b\in\mathbb C
\right\}
\subset\mathbb C^{2\times2}.
\end{equation*}
The matrix representation is used in the general derivation, whereas static
coefficients in \(\mathcal S_{dq}\) are reported below in the equivalent
complex-scalar form \(a+\mathrm{j}b\).

Under the current convention of \eqref{eq:z_parameter_two_port}, reciprocity
is $Z_{12}=Z_{21}$. Consider first the finite reciprocal $\Pi$ network of
Fig.~\ref{fig:virtual_pi_network}. Here, $Z_k^{-1}$ denotes the algebraic
rational-matrix inverse. The algebraic equivalence requires the indicated
inverses to exist; a stable proper dynamic branch realization additionally
requires its implemented inverse blocks to belong to $\mathcal{RH}_\infty$.

Let $\Sigma_Z:=Z_1+Z_2+Z_3$. Since matrices in $\mathcal S_{dq}$ commute,
inverting \eqref{eq:pi_admittance_two_port} gives, when $\Sigma_Z$ is
nonsingular,
\begin{equation}
\left\{
\begin{aligned}
 Z_{11} &= Z_1(Z_2+Z_3)\Sigma_Z^{-1},\\
 Z_{12} &= Z_1Z_3\Sigma_Z^{-1},\\
 Z_{21} &= Z_1Z_3\Sigma_Z^{-1},\\
 Z_{22} &= Z_3(Z_1+Z_2)\Sigma_Z^{-1}.
\end{aligned}
\right.
\label{eq:pi_to_z_parameters}
\end{equation}
Thus every finite reciprocal $\Pi$ network induces reciprocal impedance
parameters.

Conversely, suppose that
\[
 Z_{11},Z_{12},Z_{21},Z_{22}\in\mathcal S_{dq},
 \qquad Z_{12}=Z_{21},
\]
and define
\[
 \Delta_Z:=Z_{11}Z_{22}-Z_{12}Z_{21}.
\]
If $\Delta_Z$ is nonsingular, commutativity gives
\begin{equation}
 Z_{\mathrm{tp}}^{-1}
 =
 \Delta_Z^{-1}
 \begin{bmatrix}
  Z_{22} & -Z_{12}\\
  -Z_{12} & Z_{11}
 \end{bmatrix}.
 \label{eq:z_parameter_inverse_matrix}
\end{equation}
Comparison with \eqref{eq:pi_admittance_two_port} yields
\begin{equation}
\left\{
\begin{aligned}
 Z_1 &= \Delta_Z(Z_{22}-Z_{12})^{-1},\\
 Z_2 &= \Delta_ZZ_{12}^{-1},\\
 Z_3 &= \Delta_Z(Z_{11}-Z_{12})^{-1}.
\end{aligned}
\right.
\label{eq:z_parameters_to_pi_impedances}
\end{equation}
The finite nonsingular reciprocal $\Pi$ equivalent is therefore well defined
when
\begin{equation}
\begin{aligned}
 \det\Delta_Z&\neq0,
 &
 \det Z_{12}&\neq0,\\
 \det(Z_{22}-Z_{12})&\neq0,
 &
 \det(Z_{11}-Z_{12})&\neq0.
\end{aligned}
\label{eq:finite_pi_z_conditions}
\end{equation}
Under these conditions, the equivalent is unique because
\eqref{eq:pi_admittance_two_port} uniquely determines $Z_1^{-1}$,
$Z_2^{-1}$, and $Z_3^{-1}$.

\section{Proof of Theorem~\ref{thm:finite_gain_vi_limit}}
\label{app:proof_minimum_vi_gain}
For any feasible \(Z_v\), Weyl's inequality gives
\[
 \mu(Z+Z_v;\Omega)
 \leq \mu_0+\|Z_v\|_\infty.
\]
Hence
\[
 \|Z_v\|_\infty\geq\varepsilon-\mu_0.
\]
The constant choice
\[
 Z_v=(\varepsilon-\mu_0)I
\]
attains this bound and gives
\(\mu(Z+Z_v;\Omega)=\varepsilon\). Therefore,
\[
 \gamma_{\rm VI}^{\star}(\varepsilon)
 =\varepsilon-\mu_0.
\]

\section{Proof of Theorem~\ref{thm:reduced_gain_tp}}
\label{app:proof_reduced_gain_tp}

For the two-port in \eqref{eq:reduced_gain_tp_blocks}, \(D=0\) gives
\[
 S=[I-D(A+BZ)]^{-1}=I.
\]
The reconstruction therefore satisfies the stable-proper return condition,
and, because \(C=I\),
\[
 \widetilde Z_\rho=A+BZ
 =(\varepsilon-\rho\mu_0)I+\rho Z.
\]
Since \(\rho>0\),
\begin{align*}
 \operatorname{He}\{\widetilde Z_\rho(\mathrm{j}\omega)\}
 &=(\varepsilon-\rho\mu_0)I
 +\rho\operatorname{He}\{Z(\mathrm{j}\omega)\},\\
 \lambda_{\min}\!\left[
 \operatorname{He}\{\widetilde Z_\rho(\mathrm{j}\omega)\}\right]
 &=(\varepsilon-\rho\mu_0)
 +\rho\lambda_{\min}\!\left[
 \operatorname{He}\{Z(\mathrm{j}\omega)\}\right].
\end{align*}
Taking the infimum over \(\Omega\) gives
\(\mu(\widetilde Z_\rho;\Omega)=\varepsilon\).

The coefficient matrix is
\[
 \begin{bmatrix}
 A&B-I\\
 C-I&D
 \end{bmatrix}
 =
 \begin{bmatrix}
 (\varepsilon-\rho\mu_0)I&-(1-\rho)I\\
 0&0
 \end{bmatrix}.
\]
Its induced norm is therefore
\[
 \gamma_{\rm TP}(\rho)
 =\sqrt{(\varepsilon-\rho\mu_0)^2+(1-\rho)^2}.
\]
Finally,
\begin{align*}
 &[\gamma_{\rm VI}^{\star}(\varepsilon)]^2
 -\gamma_{\rm TP}^2(\rho)\\
 &\quad=(1-\rho)
 \left[(1+\rho)\mu_0^2-2\varepsilon\mu_0+\rho-1\right].
\end{align*}
Because \(0<\rho<1\), strict positivity is equivalent to
\eqref{eq:strict_gain_condition}.

\section{Proof of Theorem~\ref{thm:vi_family_translation}}
\label{app:proof_vi_family_translation}

For every \(p\in\mathcal P\),
\[
 \operatorname{He}\{Z(\mathrm{j}\omega,p)\}
 \succeq
 \eta(\omega,p)I
 \succeq
 \underline\eta(\omega)I.
\]
Adding \eqref{eq:general_vi_passivation_condition} gives
\[
 \operatorname{He}\{Z+Z_v\}\succeq mI,
\]
which proves VI passivation. Since the same \(Z_v\) is applied to every
\(p\), it cancels from the difference between any two family members,
proving \eqref{eq:general_vi_pairwise_invariance}.

If, in addition,
\[
 \operatorname{He}\{Z_v(\mathrm{j}\omega)\}
 =r_v(\omega)I,
 \qquad r_v(\omega)\in\mathbb R,
\]
then
\[
 \eta_{\rm VI}(\omega,p)
 =
 \eta(\omega,p)+r_v(\omega).
\]
Both endpoints of the passivity-margin interval at each frequency receive
the same scalar shift, which proves $\Delta_\mu(\widetilde{\mathcal Z}_{\rm VI};\Omega) = \Delta_\mu(\mathcal Z;\Omega)$.

\section{Proof of
Theorem~\ref{thm:cayley_robust_passivation}}
\label{app:proof_cayley_robust_passivation}

For the prescribed lower bound \(m\) and diameter bound
\(\overline\Delta\), define
\begin{equation}
 \begin{gathered}
 M:=m+\overline\Delta,\qquad
 z_s:=\sqrt{mM},\\
 \rho:=
 \frac{\sqrt{M/m}-1}{\sqrt{M/m}+1},
 \qquad 0<\rho<1.
 \end{gathered}
 \label{eq:cayley_target_parameters}
\end{equation}
These parameters satisfy
\begin{equation}
 z_s\frac{1-\rho}{1+\rho}=m,
 \qquad
 z_s\frac{1+\rho}{1-\rho}=M.
 \label{eq:cayley_boundary_relations}
\end{equation}

Let
\[
 Z_c:=Z(\cdot,p_c)
\]
and define the maximum deviation from this reference as
\begin{equation}
 \delta_Z:=
 \sup_{p\in\mathcal P}
 \|Z(\cdot,p)-Z_c\|_\infty.
 \label{eq:family_deviation_bound}
\end{equation}
To normalize the impedance family within the radius \(\rho\), choose any
\(\alpha>0\) satisfying
\[
 \alpha\geq\frac{\delta_Z}{\rho},
\]
and define
\begin{equation}
 Q:=\frac{Z-Z_c}{\alpha}.
 \label{eq:cayley_normalized_deviation}
\end{equation}
It follows that
\begin{equation}
 \|Q(\cdot,p)\|_\infty
 \leq
 \frac{\delta_Z}{\alpha}
 \leq\rho<1,
 \qquad p\in\mathcal P.
 \label{eq:cayley_normalized_bound}
\end{equation}

The inverse Cayley transformation associates this normalized radius with the
two bounds in \eqref{eq:cayley_boundary_relations}. For any \(h>0\), choose
\begin{equation}
 \begin{aligned}
 A&=h\left(I-\frac{Z_c}{\alpha}\right),&
 B&=\frac{h}{\alpha}I,\\
 C&=\frac{z_s}{2h}I,&
 D&=\frac{1}{2h}I.
 \end{aligned}
 \label{eq:cayley_twoport_realization}
\end{equation}
Because \(Z_c\in\mathcal{RH}_{\infty}^{2\times2}\), all four blocks are
stable and proper. Moreover,
\[
 H=A+BZ=h(I+Q)
\]
and
\[
 S=(I-DH)^{-1}=2(I-Q)^{-1}.
\]
The bound in \eqref{eq:cayley_normalized_bound} gives
\[
 (I-Q)^{-1}\in\mathcal{RH}_{\infty}^{2\times2}.
\]
Theorem~\ref{thm:stable_reconstruction} therefore gives a stable and proper
reconstructed terminal relation, with
\begin{equation}
 \widetilde Z_{\rm TP}
 =
 HSC
 =
 z_s(I+Q)(I-Q)^{-1}.
 \label{eq:cayley_controlled_impedance}
\end{equation}

Let
\[
 T(Q):=(I+Q)(I-Q)^{-1}.
\]
Its Hermitian part is
\[
 \operatorname{He}\{T(Q)\}
 =
 (I-Q)^{-\mathrm H}
 (I-Q^{\mathrm H}Q)
 (I-Q)^{-1}.
\]
Since \(\bar\sigma(Q)\leq\rho\),
\[
 I-Q^{\mathrm H}Q
 \succeq
 (1-\rho^2)I,
 \qquad
 \bar\sigma(I-Q)\leq1+\rho.
\]
Consequently,
\[
 \operatorname{He}\{\widetilde Z_{\rm TP}\}
 \succeq
 z_s\frac{1-\rho}{1+\rho}I
 =
 mI,
\]
where the last equality follows from
\eqref{eq:cayley_boundary_relations}. In addition,
\[
 \bar\sigma(\widetilde Z_{\rm TP})
 \leq
 z_s\frac{1+\rho}{1-\rho}
 =
 M.
\]
It follows that, for every \(p\in\mathcal P\) and
\(\omega\in\Omega\),
\[
 m
 \leq
 \eta_{\rm TP}(\omega,p)
 \leq
 \|\operatorname{He}\{\widetilde Z_{\rm TP}\}\|_2
 \leq
 \bar\sigma(\widetilde Z_{\rm TP})
 \leq
 M.
\]
Thus every passivity-margin interval lies in \([m,M]\), and its width is at
most
\[
 M-m=\overline\Delta.
\]
Taking the supremum over \(\Omega\) proves $\Delta_\mu(\widetilde{\mathcal Z}_{\rm TP};\Omega) \leq \overline\Delta$.

For Application 2, evaluation of \eqref{eq:family_deviation_bound} over
\(D_\omega/D_{\omega0}\in[0.8,1.2]\) gives
\(\alpha=539.846\) pu. Minimization of the coefficient gain over the
realization scale gives \(h=0.546835\). Substitution of these auxiliary
quantities into \eqref{eq:cayley_twoport_realization} yields the blocks in
\eqref{eq:case2_final_abcd}.

\section{Proof of
Theorem~\ref{thm:nonnegative_resistance_vm_reachability}}
\label{app:proof_nonnegative_resistance_vm_reachability}

Write \(y_k=g_k+\mathrm{j}b_k\), with \(g_k>0\). If the prescribed
\(v_m\) is not on the line through \(0\) and \(v_p\), the two branch-voltage
vectors \(v_m-v_p\) and \(v_m\) are linearly independent over the reals.
After any positive \(g_2\) and \(g_3\) are selected, the real and imaginary
parts of \eqref{eq:positive_pi_node_equation} form a nonsingular real system
for \(b_2\) and \(b_3\). Thus every noncollinear \(v_m\) is reachable. The
remaining conductance freedom can be used to avoid the singular nodal case
\[
 y_2+y_3+y_\ell=0.
\]

For a collinear point, write \(v_m=tv_p\), with \(t\in\mathbb R\). The real
part of the normalized node equation becomes
\[
 (t-1)g_2+tg_3=\phi(t),
 \qquad g_2,g_3>0.
\]
For \(0<t<1\), the two coefficients have opposite signs, so positive
\(g_2,g_3\) can realize any value of \(\phi(t)\). For \(t\leq0\), the left
side is strictly negative and feasibility is equivalent to \(\phi(t)<0\).
For \(t\geq1\), it is strictly positive and feasibility is equivalent to
\(\phi(t)>0\). These statements include the endpoints \(t=0\) and \(t=1\).
Once the conductances are selected, the imaginary part leaves one linear
equation in \(b_2,b_3\), which always has a solution. This proves
\eqref{eq:vm_nonnegative_resistance_characterization}.

\section{Network-Level Validation on the Modified IEEE 33-Bus Microgrid}
\label{app:ieee33_network_validation}

Fig.~\ref{fig:app_ieee33_topology} shows the modified IEEE 33-bus
microgrid used for the network-level tests. Eight 1-MVA VSGs are connected
at buses 6, 10, 14, 18, 22, 25, 30, and 33. Both tests use this feeder
layout, its distributed loads, and the same VSG locations. The feeder
operates in islanded mode for the operating-point uncertainty test and is
connected to the utility grid at bus 1 for the steady-state series--shunt
coordination test.

\begin{figure}[!b]
\centering
\includegraphics[width=0.98\columnwidth]
{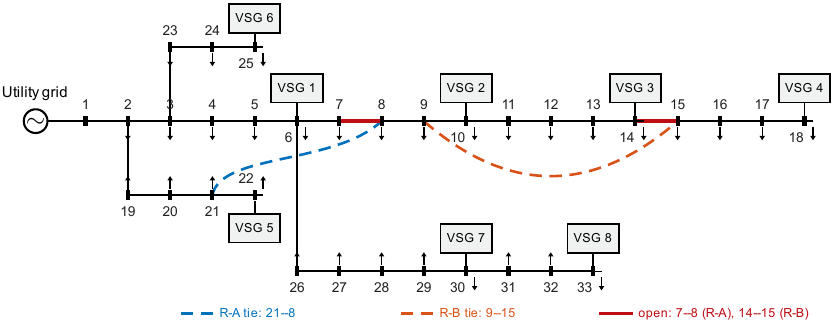}
\caption{Modified IEEE 33-bus microgrid with eight 1-MVA VSGs. The feeder
operates in islanded mode for the operating-point uncertainty test and in
grid-connected mode for the series--shunt coordination test.}
\label{fig:app_ieee33_topology}
\end{figure}

\subsection{Operating-Point Uncertainty and Bus-Voltage Response}
\label{app:ieee33_uncertainty}

The active-power operating point of each VSG varies independently over
$P_{0,k}\in[0.5,0.7]$ pu, $k=1,\ldots,8$. Each sampled operating point
defines a $2\times2$ terminal impedance $\bm Z(s,P_0)$. The fixed VI and
two-port designs are applied to all samples and yield reference
interaction-mode damping ratios of 0.78236 and 0.78125, respectively.

For $x\in\{0,\mathrm{VI},\mathrm{TP}\}$, define the complex-set diameter of
each impedance entry by
\begin{equation}
 D_{ij}^{x}(\omega)=
 \max_{p_a,p_b\in[0.5,0.7]}
 \left|Z_{ij}^{x}(\mathrm{j}\omega,p_a)
       -Z_{ij}^{x}(\mathrm{j}\omega,p_b)\right|,
 \label{eq:app_ieee33_impedance_diameter}
\end{equation}
where $ij\in\{dd,dq,qd,qq\}$. A fixed VI adds the same impedance matrix to
every member of the operating-point family and therefore leaves
$D_{ij}^{x}$ unchanged. Fig.~\ref{fig:app_ieee33_impedance_diameter} shows
that the tested two-port reduces all four diameters by approximately 67.8\%
at 0.010 Hz and maintains a comparable reduction over the tested frequency
band. The linear-fractional two-port map therefore contracts the
operating-point-dependent terminal-impedance family before it is coupled
through the microgrid.

\begin{figure}[!t]
\centering
\includegraphics[width=0.82\columnwidth]
{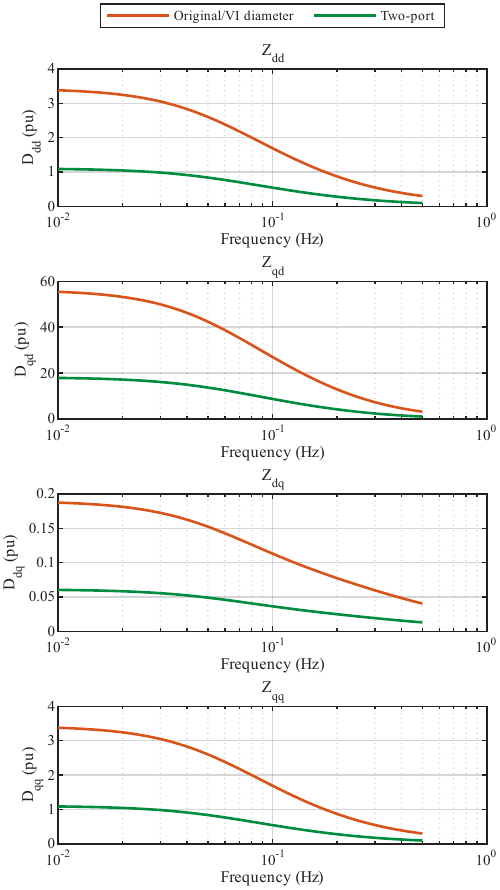}
\caption{Operating-point set diameters of the four terminal-impedance
entries. Original and VI diameters coincide because the fixed VI adds the
same impedance matrix to every member of the operating-point family.}
\label{fig:app_ieee33_impedance_diameter}
\end{figure}

For realization $r$, the closed-loop nodal admittance matrix of the
islanded microgrid is assembled as
\begin{equation}
 \bm Y_{\mathrm{cl}}(s;\bm P_0^{(r)})=
 \bm Y_{\mathrm{net},33}(s)+
 \sum_{k=1}^{8}\bm E_{b_k}
 \left[\bm Z_{\mathrm{t},k}
 (s,P_{0,k}^{(r)})\right]^{-1}
 \bm E_{b_k}^{\mathsf T},
 \label{eq:app_ieee33_network_assembly}
\end{equation}
where $\bm Y_{\mathrm{net},33}$ contains the modified IEEE 33-bus feeder
topology, branch impedances, and loads. The matrices $\bm E_{b_k}$ place
the eight $dq$ terminal admittances at buses 6, 10, 14, 18, 22, 25, 30,
and 33. As $\bm P_0^{(r)}$ changes, the resulting terminal-admittance
variations propagate through the nonuniform feeder and appear in the
nodal-voltage responses.

A sinusoidal current perturbation
\begin{equation}
 \delta i_{18,d}(t)=0.03\cos(\omega t)\ \mathrm{pu}
 \label{eq:app_ieee33_current_disturbance}
\end{equation}
is injected at bus 18. Its phasor is
$\Delta\bm i_{18}=0.03\bm e_{18,d}$, and the resulting nodal-voltage
phasor is
\begin{equation}
 \Delta\widehat{\bm v}^{(r)}(\mathrm{j}\omega)=
 \bm Y_{\mathrm{cl}}^{-1}
 (\mathrm{j}\omega;\bm P_0^{(r)})
 \Delta\bm i_{18}.
 \label{eq:app_ieee33_voltage_phasor}
\end{equation}
Here
\begin{equation*}
 \Delta\widehat{\bm v}_b^{(r)}=
 \begin{bmatrix}
 \Delta\widehat v_{b,d}^{(r)} &
 \Delta\widehat v_{b,q}^{(r)}
 \end{bmatrix}^{\mathsf T}
\end{equation*}
is the complex $dq$ voltage-response phasor at bus $b$ about the
equilibrium of realization $r$. Its magnitude is defined as
\begin{equation}
 |\Delta V_b^{(r)}(\omega)|:=
 \left\|\Delta\widehat{\bm v}_b^{(r)}
 (\mathrm{j}\omega)\right\|_2
 =
 \sqrt{
 |\Delta\widehat v_{b,d}^{(r)}|^2+
 |\Delta\widehat v_{b,q}^{(r)}|^2}.
 \label{eq:app_ieee33_voltage_response_magnitude}
\end{equation}

For the 120 sampled realizations, the pointwise envelope and its width are
\begin{equation}
 \underline g_b(\omega)
 =\min_r|\Delta V_b^{(r)}(\omega)|,
 \qquad
 \overline g_b(\omega)
 =\max_r|\Delta V_b^{(r)}(\omega)|,
 \label{eq:app_ieee33_voltage_envelope_bounds}
\end{equation}
\begin{equation}
 W_b(\omega)=
 \overline g_b(\omega)-\underline g_b(\omega).
 \label{eq:app_ieee33_voltage_envelope_width}
\end{equation}
In Fig.~\ref{fig:app_ieee33_voltage_uncertainty}, the interval
$[\underline g_b,\overline g_b]$ gives the range of the bus-voltage
response over the 120 operating-point realizations. The maximum envelope
widths at buses 6 and 33 decrease from 0.14762 and 0.14763 pu with VI to
0.007618 and 0.007598 pu with the two-port, corresponding to reductions of
94.8\% and 94.9\%, respectively. The contraction of the terminal-impedance
family is therefore reflected in a narrower range of physical bus-voltage
responses throughout the microgrid.

\begin{figure}[!t]
\centering
\includegraphics[width=0.82\columnwidth]
{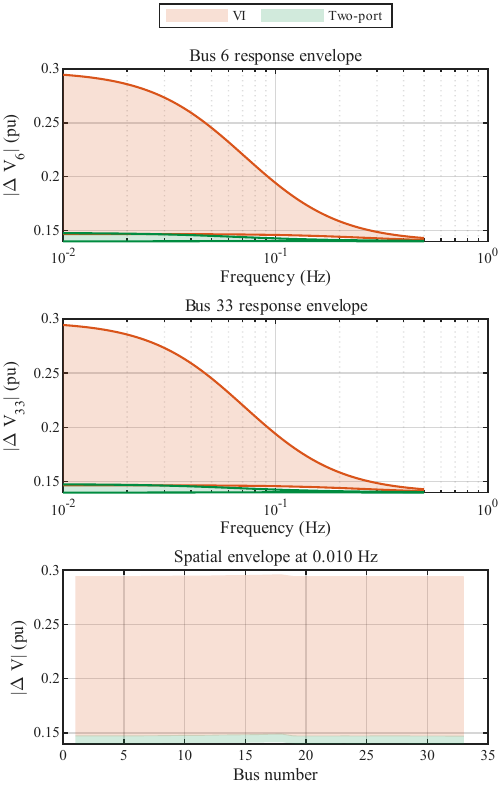}
\caption{Bus-voltage response uncertainty in the islanded eight-VSG
microgrid. Each shaded region gives the pointwise minimum-to-maximum
interval over 120 operating-point realizations. The spatial panel shows
the corresponding intervals at 0.010 Hz for all buses.}
\label{fig:app_ieee33_voltage_uncertainty}
\end{figure}

To illustrate the operating-point-dependent response in the time domain,
the $d$-axis voltage deviation at bus 18 is obtained from
\begin{equation}
 \delta v_{18,d}^{(r)}(t)=
 \Re\!\left\{
 \Delta\widehat v_{18,d}^{(r)}
 (\mathrm{j}2\pi f^\star)
 e^{\mathrm{j}2\pi f^\star t}
 \right\},
 \label{eq:app_ieee33_voltage_reconstruction}
\end{equation}
where $f^\star=0.010$ Hz. The two endpoint realizations set all eight VSGs
to $P_{0,k}=0.7$ and $0.5$ pu, respectively, and produce the minimum and
maximum VI responses at bus 33. For these realizations, the maximum
separation between the bus-18 voltage-deviation trajectories decreases
from 0.01267 pu with VI to 0.004077 pu with the two-port, a reduction of
67.8\%, as shown in Fig.~\ref{fig:app_ieee33_voltage_time}. The two-port
therefore produces a more consistent microgrid voltage response as the
eight VSG operating points vary.

\begin{figure}[!t]
\centering
\includegraphics[width=0.82\columnwidth]
{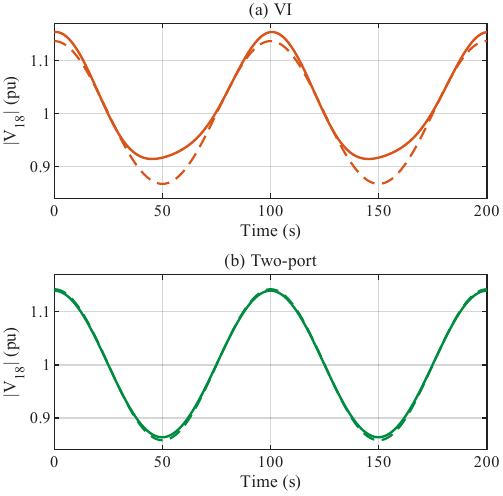}
\caption{Bus-18 $d$-axis voltage deviations under the 0.03-pu, 0.010-Hz
$d$-axis current perturbation: (a) VI and (b) two-port control. Dashed and
solid lines denote the endpoint realizations with $P_{0,k}=0.7$ and
$0.5$ pu, respectively. Both panels use the same vertical scale.}
\label{fig:app_ieee33_voltage_time}
\end{figure}

\subsection{Multi-Scenario Series--Shunt Coordination}
\label{app:ieee33_series_shunt}

In the grid-connected steady-state test, every VSG exports 0.30 MW through
an interface with $R=0.01$ pu and $X=0.06$ pu. The VI design uses series
reactances in $[0,0.18]$ pu. The reciprocal two-port uses series reactances
in the same range together with network-side shunt susceptances in
$[-0.30,0.30]$ pu.

For each strategy, one fixed parameter set is used under six operating
conditions: 0.8-pu uniform loading (L), 1.0-pu uniform loading (N), 1.2-pu
uniform loading (H), 1.0-pu loading with the loads at buses 23--33 increased
by 35\% (D), opening branch 7--8 and closing tie 21--8 (R-A), and 1.1-pu
loading with branch 14--15 open and tie 9--15 closed (R-B).

For each condition, the reactive-power spread among the eight VSGs is
defined as
\begin{equation}
 Q_{\mathrm{spr}}:=\max_k Q_k-\min_k Q_k.
 \label{eq:app_ieee33_qspread}
\end{equation}

\begin{figure}[!b]
\centering
\includegraphics[width=0.76\columnwidth]
{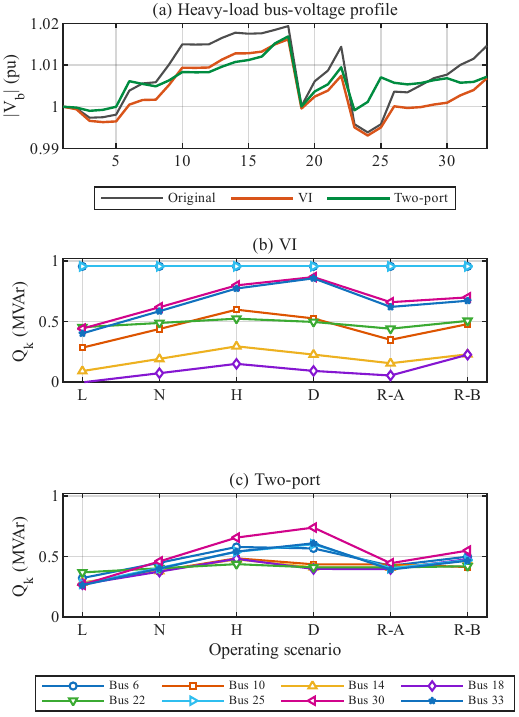}
\caption{Network-level series--shunt coordination: (a) heavy-load voltage
profiles and VSG reactive-power outputs across the six operating conditions
with (b) VI and (c) two-port control.}
\label{fig:app_ieee33_reactive_sharing}
\end{figure}

Fig.~\ref{fig:app_ieee33_reactive_sharing}(a) shows the heavy-load voltage
profiles. The two-port raises the minimum bus voltage from 0.99308 pu with VI
to 0.99901 pu. Panels (b) and (c) show the actual reactive-power outputs of
all eight VSGs across the six operating conditions. Under heavy loading, the
VI outputs span 0.1507--0.9539 MVAr, whereas the two-port outputs are confined
to 0.4367--0.6556 MVAr. The corresponding $Q_{\mathrm{spr}}$ is reduced by
72.8\%.

Across all six conditions, the reductions in $Q_{\mathrm{spr}}$ range from
60.4\% to 94.2\%. The clustering of the two-port trajectories persists under
uniform loading, downstream load concentration, and both feeder
reconfigurations. The network-side shunt channel therefore redistributes
reactive-power support across the nonuniform feeder and produces more uniform
VSG outputs than the series VI.

Under heavy loading, the two-port also reduces the maximum VSG output current
from 45.83 to 32.66 A. The associated shunt current increases the maximum
feeder current from 139.02 to 193.11 A and the network active-power loss from
0.09653 to 0.15775 MW. Thus, the improved reactive-power sharing, voltage
profile, and converter-current utilization are accompanied by higher feeder
current and network loss, which constrain practical parameter selection.

\FloatBarrier

\vspace{-9pt}
\begingroup
\renewcommand{\IEEEbibitemsep}{-0.5pt}
\bibliographystyle{IEEEtran}
\bibliography{refs}
\endgroup

\end{document}